\pdfoutput=1

\documentclass[11pt, a4paper]{article}
\usepackage[a4paper,margin=1in]{geometry}
\usepackage{amsmath,amssymb,amsfonts,mathtools,bm}
\usepackage{amsthm}
\usepackage{graphicx}
\usepackage{enumitem}
\usepackage{booktabs}
\usepackage{algorithm}
\usepackage{algpseudocode}
\usepackage[authoryear,round]{natbib}
\usepackage{hyperref}
\usepackage{microtype}
\hypersetup{colorlinks=true,linkcolor=blue,citecolor=blue,urlcolor=blue}

\newcommand{\E}{\mathbb{E}}
\newcommand{\Var}{\operatorname{Var}}
\newcommand{\Cov}{\operatorname{Cov}}
\newcommand{\NB}{\operatorname{NB}}
\newcommand{\Beta}{\operatorname{Beta}}
\newcommand{\Gam}{\operatorname{Gamma}}
\newcommand{\clip}{\operatorname{clip}}
\newcommand{\diag}{\operatorname{diag}}
\newcommand{\spr}{\rho_{\mathrm{sp}}}

\newcommand{\vecop}{\operatorname{vec}}
\newcommand{\Ne}{N_{\mathrm e}}
\newcommand{\R}{\mathbb{R}}

\theoremstyle{plain}
\newtheorem{proposition}{Proposition}
\theoremstyle{remark}

\title{Negative-binomial ensemble filtering for discrete-time Hawkes network reconstruction}
\author{Naratip Santitissadeekorn\\[0.4em]
{\small University of Surrey, Guildford, United Kingdom}\\
{\small\texttt{n.santitissadeekorn@surrey.ac.uk}}}
\date{}

\begin{document}
\maketitle

\begin{abstract}
We develop a sequential ensemble filter for reconstructing directed excitation networks from overdispersed multivariate counts observed on a regular time grid. The data-generating model is a discrete-time negative-binomial Hawkes process, for which we prove exact first- and second-moment stability conditions and derive a row-wise sufficient condition. The filter learns the network and the dispersion online, and a row-composition parameter forecast keeps enough ensemble variation to respond to structural change. In synthetic experiments, we demonstrate that a Poisson-based filter provides comparable point-estimate accuracy of the network, while the proposed filter provides substantially better calibrated predictive and parameter uncertainty. Under an abrupt redistribution of influence, the filter recovers connections that are absent before the change. An application to daily property-crime counts in sixteen Chicago community areas gives broadly calibrated one-step predictions and a recovered network consistent with independent temporal and geographic information.
\end{abstract}


\section{Introduction}\label{sec:intro}

Multivariate event streams in which the timing and type of events are informative about future activity arise in financial markets \citep{hawkes2018review,bacry2015finance}, epidemiology \citep{meyer2012space}, seismology \citep{ogata1988earthquake}, crime \citep{mohler2011crime,santitissadeekorn2018selfexciting,santitissadeekorn2020poisson} and service networks \citep{daw2018queues,koops2018infinite}. Hawkes processes \citep{hawkes1971point,hawkes1971spectra,hawkes1974cluster} provide the standard framework. In the multivariate case the excitation coefficients form a weighted directed network: the diagonal entries describe self-excitation and the off-diagonal entries describe which component tends to excite subsequent activity in which other component and how strongly. It is also known that for linear Hawkes processes observed through event times the absence of an excitation kernel from \(j\) to \(i\) coincides with the absence of Granger causality \citep{eichler2017graphical,xu2016learning,achab2017uncovering}.

The present paper concerns a different observational regime of overdispersed counts on a regular grid rather than event times. One approach retains a latent continuous-time Hawkes process as the inferential target and treats the counts as aggregated observations of it. The INAR\((\infty)\) construction \citep{kirchner2016inar} converges weakly to a Hawkes process as the mesh shrinks, and the associated INAR\((p)\) estimator rescales conditional-least-squares coefficients to recover the baseline and kernel \citep{kirchner2017estimation}. At fixed bin width both discretisation and lag truncation introduce error. Binned likelihoods and INAR approximations can be biased and highly variable as the bin widens, and an EM scheme over latent within-bin event times has been proposed instead \citep{shlomovich2022binned}. A sequential Monte Carlo route to the same incomplete-data likelihood has also been developed \citep{chen2025discrete}. We take a different position. The discrete-time negative-binomial Hawkes recursion introduced in Section~\ref{sec:model} is treated as the data-generating model itself. No within-bin event times or continuous-time intensity path are posited or reconstructed. The inferential target is the excitation matrix of the discrete-time count model, and the method is designed to reconstruct that matrix sequentially and to adapt when it changes.

The closest predecessor to the proposed filter is the ensemble Poisson--gamma filter (EnPGF) for Hawkes network reconstruction from count data \citep{santitissadeekorn2023physd,santitissadeekorn2020poisson}, which belongs to the broader tradition of ensemble Kalman filtering \citep[EnKF,][]{evensen1994sequential} and its distribution-specific non-Gaussian variants such as the GIGG-EnKF \citep{bishop2016gigg}. Note that the EnPGF assumes conditionally Poisson counts, whereas observed counts often contain additional variation from unmodelled heterogeneity, local bursts, imperfect binning or unresolved within-step dynamics. A conditionally Poisson self-exciting process is already marginally overdispersed because its conditional mean is random. Thus, the relevant distinction is conditional variation. Under a correctly specified conditional mean, the Poisson score remains a valid estimating function and the Poisson quasi-likelihood estimator retains the same probability limit \citep{fokianos2009poisson,christou2014quasi,ahmad2016qmle,davis2016theory}, but its model-based precision is too high, since a Poisson likelihood attributes more information to each count than an overdispersed count contains. This is consequential in sequential estimation because ensemble spread controls the next analysis increment. Overstating the information in each count contracts the parameter ensemble too quickly and eventually weakens its response to new information or structural change. We therefore use a negative-binomial observation model to represent the information in the counts through the ensemble spread. Related negative-binomial mechanisms for self-exciting counts appear in the INGARCH literature \citep{ferland2006integer,zhu2011negative}, of which the present conditional-mean recursion is a multivariate linear instance, and in self-exciting negative-binomial distribution processes \citep{hisakado2022multiterm,hisakado2022multidimensional,sakuraba2024self}, which excite both mean and dispersion. Multivariate count autoregressions that contain this recursion as a special case have been studied with conditionally Poisson marginals and contemporaneous dependence between nodes \citep{fokianos2020multivariate} and with conditionally independent counts whose variance may be quadratic in the mean \citep{bracher2017periodic}. Section~\ref{sec:stability} compares their assumptions with those made here.

The paper is organised as follows. Section~\ref{sec:model} specifies the model, its parameter domain and its moment stability. Section~\ref{sec:filter} develops the filter and states the complete algorithm. Section~\ref{sec:experiments} reports twin experiments on static reconstruction and network-change tracking, together with a paired comparison against the ensemble Poisson--gamma filter of the predecessor work. Section~\ref{sec:chicago} applies the filter to daily property-crime counts in sixteen Chicago community areas over 2021--2024. The deterministic seasonal cycle is estimated offline, the filter is run on the anomaly and the dispersion is estimated node by node. Appendix~\ref{app:stability_proof} proves the stability proposition.

\section{Discrete-time negative-binomial Hawkes model}\label{sec:model}

\subsection{Model specification and parameters}\label{sec:model_spec}

Let \(Y_{k,i}\in\{0,1,2,\ldots\}\) be the number of events recorded at node \(i\in\{1,\ldots,m\}\) during time step \(k\) and let \(\Delta t>0\) be the observation interval. The discrete-time conditional rate \(\lambda_{k,i}\) and conditional mean count \(M_{k,i}=\lambda_{k,i}\Delta t\) evolve by the multivariate Hawkes recursion with exponential memory,
\begin{equation}
\lambda_{k+1,i}
=
\mu_i+
(1-\beta_i\Delta t)(\lambda_{k,i}-\mu_i)
+
\sum_{j=1}^{m}\alpha_{ij}Y_{k,j},
\qquad i=1,\ldots,m ,
\label{eq:model}
\end{equation}
where \(\mu_i\) is a baseline rate, \(\beta_i\) a decay rate and \(\alpha_{ij}\) the excitation coefficient from node \(j\) to node \(i\). The matrix \(A=(\alpha_{ij})\) is a weighted directed network. Since the observations are counts on a grid rather than event times, \(A\) is interpreted as a predictive excitation network. A nonzero \(\alpha_{ij}\) means that past counts at \(j\) contribute to the conditional rate of \(i\) within \eqref{eq:model}.

Conditional on \(M_{k,i}\) and a dispersion parameter \(r>0\), counts follow a negative-binomial law,
\begin{equation}
Y_{k,i}\mid M_{k,i},r\sim\NB\Big(r,\frac{r}{r+M_{k,i}}\Big),
\qquad
\E(Y\mid M,r)=M,\qquad \Var(Y\mid M,r)=M+\frac{M^2}{r},
\label{eq:nb}
\end{equation}
Large \(r\) approaches the Poisson limit, while small \(r\) permits strong conditional variability. Equivalently, \(Y_{k,i}\) is Poisson with random mean \(M_{k,i}\varepsilon_{k,i}\) for independent gamma variables \(\varepsilon_{k,i}\) with shape \(r\) and scale \(1/r\), which is how counts are simulated in Section~\ref{sec:experiments}. Conditional on the intensity vector the observation model factorises across nodes.

The recursion is a rate model only if its persistence coefficient is nonnegative, and hence we impose
\begin{equation}
\mu_i>0,\qquad 0<\beta_i\Delta t\leq1,\qquad \alpha_{ij}\geq0 .
\label{eq:domain}
\end{equation}
Under \eqref{eq:domain}, \eqref{eq:model} is a nonnegative combination of nonnegative terms and gives \(\lambda_{k+1,i}\geq\beta_i\Delta t\,\mu_i>0\) whenever \(\lambda_{k,i}\geq0\). The upper bound on \(\beta_i\Delta t\) is necessary. Without it the conditional mean can leave the positive half-line and the transform on which Section~\ref{sec:filter} is built is undefined. The inferential target is the parameter vector \(\theta=(\theta_1^\top,\ldots,\theta_m^\top)^\top\) with \(\theta_i=(\mu_i,\beta_i,\alpha_{i1},\ldots,\alpha_{im})^\top\), and in particular the excitation matrix \(A\), together with the common dispersion \(r\) when it is unknown.

Two properties of \eqref{eq:nb} are used later. The per-observation log-likelihood is
\begin{equation}
\begin{split}
\ell_{\mathrm{NB}}(M;y,r)
&=\log\Gamma(y+r)-\log\Gamma(r)-\log\Gamma(y+1)\\
&\qquad+r\log r+y\log M-(y+r)\log(r+M),
\end{split}
\label{eq:nb_loglik}
\end{equation}
and its score and Fisher information with respect to the mean are
\begin{equation}
\frac{\partial\ell_{\mathrm{NB}}}{\partial M}=\frac{r\,(y-M)}{M(r+M)}=\frac{r}{r+M}\cdot\frac{y-M}{M},
\qquad
\mathcal I_{\mathrm{NB}}(M)=\frac{r}{M(r+M)}=\frac{r}{r+M}\,\mathcal I_{\mathrm P}(M),
\label{eq:nb_score}
\end{equation}
where \((y-M)/M\) and \(\mathcal I_{\mathrm P}(M)=1/M\) are the Poisson score and information. A negative-binomial count therefore carries the fraction \(r/(r+M)\) of the information in a Poisson count with the same conditional mean.

\subsection{Moment stability of the generative model}\label{sec:stability}

This section discusses the stability conditions under which the first- and second-moment recursions of the count model \eqref{eq:model}--\eqref{eq:nb} converge. The conditions are used to select simulated truths with finite stationary variance and to characterise the parameter forecast of Section~\ref{sec:adaptability}.

In terms of the count-mean vector \(\bm M_k=(M_{k,1},\ldots,M_{k,m})^\top\), \eqref{eq:model} becomes
\begin{equation}
\bm M_{k+1}=\bm\omega+B\bm M_k+\mathcal A\,\bm Y_k,
\qquad
\omega_i=\mu_i\beta_i\Delta t^2,\quad B=\diag(1-\beta_i\Delta t),\quad \mathcal A=\Delta t\,A,
\label{eq:count_mean}
\end{equation}
and we write \(G=B+\mathcal A\), which is entrywise nonnegative under \eqref{eq:domain}.

At the level of the conditional mean, the recursion \eqref{eq:count_mean} is a time-homogeneous first-order member of two classes of multivariate linear count autoregressions \citep{fokianos2020multivariate,bracher2017periodic}. In both, the matrix multiplying the lagged conditional mean need not be diagonal, so the conditional mean of one node can pass directly into that of another. Here \(B\) is diagonal. Each node relaxes towards its own baseline and nodes interact only through the counts. The direct contribution of \(Y_{k,j}\) to \(M_{k+\ell+1,i}\), carried by the intensity of node \(i\), is \((1-\beta_i\Delta t)^{\ell}\Delta t\,\alpha_{ij}\). Therefore, it decays at the rate of the receiving node as in \eqref{eq:model}. The two classes differ in the conditional law of \(\bm Y_k\). In the first class \citep{fokianos2020multivariate}, each count is Poisson given the past, and the counts of different nodes in the same step are linked by an unspecified copula placed on the waiting times of the underlying Poisson processes. That construction also admits mixed-Poisson marginals, including the gamma-mixed Poisson form of \eqref{eq:nb}. Here, as in the second class \citep{bracher2017periodic}, the counts are conditionally independent given \(\bm M_k\), and each is negative binomial as in the endemic--epidemic models for multivariate infectious disease surveillance counts \citep{held2005statistical,meyer2014power}. The joint conditional law is then the product of the marginals in \eqref{eq:nb}.

These assumptions have different consequences for estimation and for the moments. The Poisson quasi-likelihood estimator for the first class \citep{fokianos2020multivariate} treats the nodes as independent within a step. Under the stability and regularity conditions given there, correct specification of the conditional mean is enough for it to be consistent and asymptotically normal, and the copula does not need to be estimated. The copula affects only its asymptotic covariance. Related quasi-likelihood results for Poisson network autoregressions with a known network are also available \citep{armillotta2024count}. The analysis update of Section~\ref{sec:analysis} instead takes the product of the negative-binomial factors as the joint observation likelihood since the counts are conditionally independent. Without conditional independence the same scalar updates remain usable under working independence, but the ensemble then represents the working model rather than the joint filtering distribution. For Poisson counts, the contraction condition under which a stationary solution exists also gives finite moments of all orders \citep[Proposition~3.2]{fokianos2020multivariate}. For negative-binomial counts the conditional variance is quadratic in the conditional mean, and with the count fed back into the rate, mean stability no longer implies finite variance.

\begin{proposition}[Moment stability]\label{prop:stability}
Consider \eqref{eq:count_mean} under \eqref{eq:nb}, assume \eqref{eq:domain} and \(\bm\omega>0\), and let \(E=[e_1\otimes e_1,\ldots,e_m\otimes e_m]\in\R^{m^2\times m}\), where \(e_i\) is the \(i\)th standard basis vector of \(\R^m\).
\begin{enumerate}[label=(\roman*),leftmargin=*]
\item The mean sequence \(\E[\bm M_k]\) converges from every finite initialisation if and only if \(\spr(G)<1\), with limit \((I-G)^{-1}\bm\omega\), and any strictly stationary solution with finite mean satisfies \(\spr(G)<1\). Since \(G\geq0\), the row-sum bound
\begin{equation}
\max_{i}\Big(1-\beta_i\Delta t+\Delta t\sum_{j=1}^{m}\alpha_{ij}\Big)<1
\label{eq:first_order}
\end{equation}
is sufficient.
\item The covariance sequence \(\Cov(\bm M_k)\) converges from every finite-covariance initialisation if and only if
\begin{equation}
\spr(L)<1,\qquad L=G\otimes G+\frac1r\,(\mathcal A\otimes\mathcal A)EE^\top ,
\label{eq:second_order_exact}
\end{equation}
diverging in Frobenius norm otherwise, and any strictly stationary solution with finite covariance satisfies \eqref{eq:second_order_exact}. A sufficient per-row test is
\begin{equation}
\max_{i}\left\{\Big(1-\beta_i\Delta t+\Delta t\sum_{j}\alpha_{ij}\Big)^{2}+\frac{\Delta t^{2}}{r}\Big(\sum_{j}\alpha_{ij}\Big)^{2}\right\}<1 ,
\label{eq:second_order_row}
\end{equation}
which is exact when \(m=1\), where \(\spr(L)=\gamma^2+a^2/r\) with \(\gamma=1-\beta\Delta t+\alpha\Delta t\) and \(a=\alpha\Delta t\).
\end{enumerate}
\end{proposition}

The proof is given in Appendix~\ref{app:stability_proof}. For \(m=1\), Proposition~\ref{prop:stability}(ii) reduces to \(\gamma^2+a^2/r<1\), which is the second-order stationarity condition of the negative-binomial INGARCH(1,1) model \citep[Remark~3]{zhu2011negative} and the time-homogeneous case of the corresponding condition for periodically stationary models \citep{bracher2017periodic}. With the conditional mean and variance of \eqref{eq:nb}, \(1-\gamma^2-a^2/r\) is the denominator of the stationary variance of the same model \citep{christou2014quasi}. For multivariate models of this type, the stationary second moments have been obtained by iterating the moment recursions under the assumption that the process is second-order stationary \citep{bracher2017periodic}. Proposition~\ref{prop:stability}(ii) gives the condition under which the second-moment recursion of the network model converges, and the recursion diverges otherwise. The second-moment conditions \eqref{eq:second_order_exact}--\eqref{eq:second_order_row} tighten as \(r\) decreases and couple \(A\) to \(r\), whereas the mean-stability condition \(\spr(G)<1\) is independent of \(r\). Thus, a network can be stable in mean while its covariance diverges in a strongly overdispersed regime.

\section{Negative-binomial ensemble filter}\label{sec:filter}

The general mechanism of the proposed filter can be briefly described as follows. An ensemble of \(\Ne\) members \((\lambda_k^{(s)},\theta_k^{(s)})\), \(s=1,\ldots,\Ne\), is used to approximate the uncertainty in the intensity and the parameters. Each member carries its own intensity vector and its own parameter values, and the filter revises every member as each count vector arrives. A time step consists of a forecast and an analysis. In the forecast each member is moved one step ahead under the model with its own parameters, which gives the ensemble's prediction of the next count vector. In the analysis the count vector that has just been observed is compared with that prediction and each member is corrected towards it in two stages, first in its intensity and then in its parameters. Quantities before the update carry the superscript \(f\) for forecast and quantities after it the superscript \(a\) for analysis. Thus, \(\lambda^{f,(s)}\) is a member's predicted intensity and \(\lambda^{a,(s)}\) its corrected value. A superscript \(\dagger\) denotes the data-generating value in simulations. The true network is written \(A^\dagger\) when it is constant and \(A_k^\dagger\) when it is piecewise constant.

In the forecast the parameter members are first perturbed slightly to replenish the spread that repeated updates remove and to keep the ensemble able to respond to changes in the network. Section~\ref{sec:adaptability} specifies these perturbations for \(\mu\), \(\beta\) and \(A\). The intensity is then propagated through \eqref{eq:model} with each member's forecast parameters and the observed counts,
\begin{equation}
\lambda_{k+1}^{f,(s)}=\mu_{k+1}^{f,(s)}+\big(1-\beta_{k+1}^{f,(s)}\Delta t\big)\big(\lambda_k^{a,(s)}-\mu_{k+1}^{f,(s)}\big)+A_{k+1}^{f,(s)}Y_k ,
\label{eq:intensity_forecast}
\end{equation}
so the observed history is common to all members while their intensities differ through the parameter ensemble.

\subsection{Parameter forecast}\label{sec:adaptability}

Repeated analysis updates of a nominally fixed parameter would reduce the differences between members so that the empirical parameter variance severely understates the remaining uncertainty. This finite-ensemble underdispersion, or inbreeding, is a common issue in EnKF and is especially consequential for parameters, which have no forecast dynamics to replenish lost spread \citep{houtekamer1998enkf,ruckstuhl2018parameters}. Thus, the parameter ensemble has to be perturbed at each forecast step. This strategy has long been used for static-parameter learning in particle filters \citep{liu2001combined}, and Section~\ref{sec:regression} shows that without it the regression update cannot move a parameter whose ensemble spread has collapsed.

The simplest forecast for a nonnegative parameter is to multiply the value held by each member by an independent gamma variate with mean one. For the baseline rate \(\mu_i\), the decay rate \(\beta_i\) and the excitation coefficients \(\alpha_{ij}\) of member \(s\), this gives
\begin{equation}
\begin{gathered}
\mu_{k+1,i}^{f,(s)}=\zeta_{k,i}^{\mu,(s)}\,\mu_{k,i}^{a,(s)},\qquad
\beta_{k+1,i}^{f,(s)}=\zeta_{k,i}^{\beta,(s)}\,\beta_{k,i}^{a,(s)},\qquad
\alpha_{k+1,ij}^{f,(s)}=\zeta_{k,ij}^{\alpha,(s)}\,\alpha_{k,ij}^{a,(s)},\\
\zeta_{k,i}^{\mu,(s)},\ \zeta_{k,i}^{\beta,(s)},\ \zeta_{k,ij}^{\alpha,(s)}\sim\Gam\big(\text{shape }\sigma_\zeta^{-2},\ \text{scale }\sigma_\zeta^{2}\big),
\end{gathered}
\label{eq:jitter}
\end{equation}
with all multipliers drawn independently across parameters, members and steps. Each multiplier has mean one and relative standard deviation \(\sigma_\zeta\). Thus, this ``jitter'' keeps every parameter nonnegative and leaves the ensemble mean unchanged in expectation. Since a multiplier above one can carry \(\beta_i\Delta t\) past one, the jittered members have to be projected onto \eqref{eq:domain} before the intensity is propagated. Note also that since the perturbation of each member is proportional to its current value, the jitter restores almost no forecast variance at an entry of \(A\) that the analysis updates have reduced to near zero. This becomes an issue when the network changes because such an entry can then remain near zero.

For the matrix \(A\) the composition forecast is an alternative that may avoid this issue. It is designed for changes that redistribute influence among source nodes. For member \(s\) and destination row \(i\), let \(S_{k,i}^{(s)}=\sum_j\alpha_{k,ij}^{a,(s)}\) be the analysed row total and \(p_{k,ij}^{a,(s)}=\alpha_{k,ij}^{a,(s)}/S_{k,i}^{(s)}\) its composition. Before renormalisation a probability floor \(p_{\min}\) keeps every edge strictly positive after an analysed value approaches zero. If \(S_{k,i}^{(s)}=0\) the composition is set to \(1/m\) by convention while the row remains at zero. The composition is floored, diffused in centred log-ratio coordinates and mapped back by the softmax at fixed total, i.e.,
\begin{equation}
\begin{aligned}
\tilde p_{k,ij}^{a,(s)}&=\frac{\max\{p_{k,ij}^{a,(s)},p_{\min}\}}{\sum_{\ell}\max\{p_{k,i\ell}^{a,(s)},p_{\min}\}},\\
z_{k,ij}^{a,(s)}&=\log\tilde p_{k,ij}^{a,(s)}-\frac1m\sum_{\ell}\log\tilde p_{k,i\ell}^{a,(s)},\\
z_{k+1,ij}^{f,(s)}&=z_{k,ij}^{a,(s)}+\xi_{k,ij}^{(s)}-\frac1m\sum_{\ell}\xi_{k,i\ell}^{(s)},
\qquad \xi_{k,ij}^{(s)}\sim\mathcal N(0,q_A^2),\\
p_{k+1,ij}^{f,(s)}&=\frac{\exp(z_{k+1,ij}^{f,(s)})}{\sum_{\ell}\exp(z_{k+1,i\ell}^{f,(s)})},
\qquad
\alpha_{k+1,ij}^{f,(s)}=S_{k,i}^{(s)}\,p_{k+1,ij}^{f,(s)} .
\end{aligned}
\label{eq:composition}
\end{equation}

The centred log-ratio map followed by the softmax is the identity on the simplex, and the perturbation stays in the sum-zero subspace. Thus, with \(q_A=0\) and \(p_{\min}=0\) the construction reproduces the analysed composition. With \(p_{\min}>0\) it reproduces the floored and renormalised composition.
The composition forecast maintains adaptability through the probability floor \(p_{\min}\), which keeps every edge in the interior of the simplex, and the centred log-ratio perturbation with scale \(q_A\), which supplies forecast variation while preserving each member's row total. The parameters \(\mu\) and \(\beta\) are jittered by \eqref{eq:jitter} in either case, and the two forecasts for \(A\) are compared in the network-change experiment of Section~\ref{sec:tracking_results}.

Holding the row total fixed preserves the row-sum bound \eqref{eq:first_order} and the row-wise second-moment bound \eqref{eq:second_order_row} for fixed \(\beta\) and \(r\). The exact spectral radius \(\spr(B+\mathcal A)\) can still change as mass moves within a row. The floor \(p_{\min}\) moves a previously absent edge into the interior of the simplex, where \(q_A\) supplies forecast variation. Without the floor the absolute perturbation of a small edge remains of order \(q_A\alpha_{ij}\). Although the perturbation has mean zero in log-ratio coordinates, the nonlinear softmax does not preserve the ensemble mean on the simplex. Its leading mean effect is of order \(q_A^2\), and repeated diffusion without analysis updates tends to flatten the ensemble-mean composition toward uniformity. In the filter this tendency is opposed by the analysis update. Most importantly for adaptation, the forecast replenishes \(\Var(\alpha_{ij}^f)\) which allows informative lagged counts at source node \(j\) to restore \(\Cov(\alpha_{ij}^f,M^f)\).

\subsection{Stage 1: Negative-binomial analysis update}\label{sec:analysis}

Let \(\theta_k=(\mu_k,\beta_k,A_k)\), let \(\mathcal Y_k=\{Y_1,\ldots,Y_k\}\), and denote the filtering target by \(p(\lambda_k,\theta_k\mid\mathcal Y_k)\). The forecast ensemble \((\lambda_k^{f,(s)},\theta_k^{f,(s)})\) is converted to an analysis ensemble \((\lambda_k^{a,(s)},\theta_k^{a,(s)})\) in two stages. The intensity is updated first followed by a regression update of the parameters. Since the observation model factorises across nodes, Stage 1 is described for one node with the index \(i\) suppressed. Dependence across nodes enters through the forecast and through the regression of Stage 2 in Section~\ref{sec:regression}.

For fixed \(r\) put \(u=M/(r+M)\in(0,1)\) with inverse \(M=ru/(1-u)\). In terms of \(u\) the negative-binomial likelihood is proportional to \(u^{y}(1-u)^{r}\), a beta kernel, so if the forecast law of \(u\) is approximated by \(\Beta(a_f,b_f)\), the posterior after observing \(y\) is
\begin{equation}
u^{a}\mid y\sim\Beta(a_f+y,\;b_f+r).
\label{eq:beta_update}
\end{equation}
The beta parameters are obtained by moment matching the transformed forecast ensemble \(u^{f,(s)}=M^{f,(s)}/(r+M^{f,(s)})\), where \(M^{f,(s)}=\lambda^{f,(s)}\Delta t\). With \(\bar u^f\) and \(v_u^f\) the ensemble mean and variance, \(\kappa_f=\bar u^f(1-\bar u^f)/v_u^f-1\), \(a_f=\kappa_f\bar u^f\), and \(b_f=\kappa_f(1-\bar u^f)\), after truncating \(\bar u^f\) and \(v_u^f\) away from degenerate values so that \(a_f,b_f>0\). Conjugacy then gives \eqref{eq:beta_update} exactly for fixed \(r\).

The inverse transformation imposes a constraint on the update. Under it, \(M/r\) has a beta-prime distribution with parameters \((a,b)\), whose mean is finite only if \(b>1\) and whose variance is finite only if \(b>2\). The analysed count mean therefore has finite variance only when
\begin{equation}
b_f+r>2 .
\label{eq:bf_condition}
\end{equation}
This constraint is operative because the memberwise increment of the analysed count mean enters the regression update of Section~\ref{sec:regression}. Flooring \(b_f\) alone would shift the mean \(a_f/(a_f+b_f)\) of the approximation. We instead apply the floor through the concentration parameter. With a small tolerance \(\epsilon_b>0\) and \(b_\star(r)=\max\{\epsilon_b,\,2-r+\epsilon_b\}\),
\begin{equation}
\kappa_f=\max\Big\{\frac{\bar u^f(1-\bar u^f)}{v_u^f}-1,\ \frac{b_\star(r)}{1-\bar u^f}\Big\},
\qquad a_f=\kappa_f\bar u^f,\qquad b_f=\kappa_f(1-\bar u^f),
\label{eq:concentration_floor}
\end{equation}
which gives \(b_f+r\geq2+\epsilon_b\) while leaving the beta mean equal to \(\bar u^f\). The value \(b_\star\) depends on the current dispersion estimate and is recomputed whenever \(r\) is updated. When the floor is active, the beta approximation has a smaller variance than the forecast ensemble, and the analysis understates the forecast uncertainty.

The parameter regression uses the memberwise increment \(\Delta M^{(s)}=M^{a,(s)}-M^{f,(s)}\), so the posterior draws must be paired with the forecast members. We draw \(\widetilde u^{a,(s)}\sim\Beta(a_f+y,\,b_f+r)\), sort the draws, and assign them by forecast rank. If \(\pi_f\) orders the forecast ensemble, \(u^{f,(\pi_f(1))}\leq\cdots\leq u^{f,(\pi_f(\Ne))}\), and \(\widetilde u^{a,(1)}\leq\cdots\leq\widetilde u^{a,(\Ne)}\) are the sorted draws, then
\begin{equation}
u^{a,(\pi_f(q))}=\widetilde u^{a,(q)},\qquad q=1,\ldots,\Ne .
\label{eq:rank}
\end{equation}
Equation~\eqref{eq:rank} pairs low forecast values with low posterior draws and high values with high draws. This retains the ordering on which the regression transform relies. It also minimises the total squared memberwise increment among all assignments of the same posterior sample by the rearrangement inequality \citep{hardy1952inequalities}. This ordered assignment is motivated by the two-stage ensemble construction for an EnKF implementation \citep{anderson2003local}. Sampling represents the beta posterior on its bounded support, including its skewness, while the ordered assignment provides the memberwise coupling required by the regression update. Finally, the analysed ensemble is mapped back by \(M^{a,(s)}=ru^{a,(s)}/(1-u^{a,(s)})\) and \(\lambda^{a,(s)}=M^{a,(s)}/\Delta t\). The update requires two sorts per node and step with cost \(O(\Ne\log\Ne)\).

\subsection{Stage 2: Regression update of the parameters}\label{sec:regression}

For a generic parameter \(q\in\{\mu_i,\beta_i,\alpha_{ij}\}\) of row \(i\), the intensity increment of node \(i\) is transformed into a parameter increment by the finite-ensemble linear regression
\begin{equation}
q^{a,(s)}
=
q^{f,(s)}
+
\frac{\Cov(q^f,M^f)}{\Var(M^f)}
\left(M^{a,(s)}-M^{f,(s)}\right),
\label{eq:regression}
\end{equation}
with covariances and variances taken across the ensemble. This is an approximate linear transform of the intensity increment into the conditional parameter update induced by the analysed count mean, which discards nonlinear dependence.

The regression coefficient \(K_{\alpha,ij}=\Cov(\alpha_{ij}^f,M^f)/\Var(M^f)\) of an edge depends on the ensemble covariance between the coefficient and the forecast count mean. The finite-ensemble underdispersion described in Section~\ref{sec:adaptability} shrinks this covariance. By the Cauchy--Schwarz inequality \(|\Cov(\alpha_{ij}^f,M^f)|\leq\{\Var(\alpha_{ij}^f)\Var(M^f)\}^{1/2}\). It follows that if \(\Var(M^f)\geq d_0>0\) and the innovation satisfies \(|M^{a,(s)}-M^{f,(s)}|\leq C\), the memberwise increment obeys
\begin{equation}
\left|K_{\alpha,ij}\left(M^{a,(s)}-M^{f,(s)}\right)\right|
\leq
C\left\{\frac{\Var(\alpha_{ij}^{f})}{\Var(M^f)}\right\}^{1/2}
\leq
C\left\{\frac{\Var(\alpha_{ij}^{f})}{d_0}\right\}^{1/2}
\longrightarrow 0
\qquad\text{as }\Var(\alpha_{ij}^f)\to0 .
\label{eq:collapse}
\end{equation}
This establishes the link between ensemble collapse and loss of adaptability. The parameter update becomes small whatever the count innovation within the bound, so the regression update can move \(\alpha_{ij}\) again only after the parameter forecast has restored \(\Var(\alpha_{ij}^f)\).

After this update each ensemble member is projected onto the parameter domain \eqref{eq:domain}. The projection is required because a member with a nonpositive conditional mean makes the transform \(u\) undefined. The projection is one-sided for \(\alpha_{ij}\) and therefore biases the ensemble mean upward for edges whose true value is zero or near zero.

\subsection{Online estimation of dispersion}\label{sec:dispersion}

When \(r\) is unknown, its update is delayed and damped because early residuals reflect both overdispersion and inaccurate network estimates. The following analysis uses one common \(r\) for all nodes but it can be readily adapted to a node-dependent \(r\). For \(h\geq1\) the lag-\(h\) autocovariance of the counts is \(G^{h-1}\Cov(\bm M_{k+1},\bm Y_k)\). Therefore, the rate at which it decays is fixed by the mean recursion and does not involve \(r\), which enters only through the level of the second moments.

The identity \(\Var(Y\mid M,r)-M=M^2/r\) suggests replacing \(M\) by its one-step forecast mean and forming a ratio of sums, but keeping in mind that the forecast mean is not the true conditional mean. Writing \(\bar M^f_{\ell,i}\) and \(V^f_{\ell,i}\) for the ensemble mean and variance of the forecast count mean, let \(\mathcal F_{\ell-1}\) denote the information set generated by \(\mathcal Y_{\ell-1}\) together with the filter's own randomisation up to the forecast at step \(\ell\) so that the forecast ensemble is \(\mathcal F_{\ell-1}\)-measurable although not \(\mathcal Y_{\ell-1}\)-measurable. Since \(\E[Y_{\ell,i}\mid M_{\ell,i}]=M_{\ell,i}\) and \(\bar M^f_{\ell,i}\) is \(\mathcal F_{\ell-1}\)-measurable, the expected squared residual is exactly
\begin{equation}
\E\!\left[(Y_{\ell,i}-\bar M^f_{\ell,i})^2\,\middle|\,\mathcal F_{\ell-1}\right]
=
\E\!\left[\Var(Y_{\ell,i}\mid M_{\ell,i})\,\middle|\,\mathcal F_{\ell-1}\right]
+
\E\!\left[(M_{\ell,i}-\bar M^f_{\ell,i})^2\,\middle|\,\mathcal F_{\ell-1}\right],
\label{eq:total_variance}
\end{equation}
in which the first term is \(\E[M_{\ell,i}\mid\mathcal F_{\ell-1}]+\E[M_{\ell,i}^2\mid\mathcal F_{\ell-1}]/r\) by \eqref{eq:nb}. The predictive mean and variance of \(M_{\ell,i}\) are not available, and replacing them by the ensemble mean and variance, that is \(\E[M_{\ell,i}\mid\mathcal F_{\ell-1}]\approx\bar M^f_{\ell,i}\) and \(\E[(M_{\ell,i}-\bar M^f_{\ell,i})^2\mid\mathcal F_{\ell-1}]\approx V^f_{\ell,i}\), gives
\begin{equation}
\E\!\left[(Y_{\ell,i}-\bar M_{\ell,i}^f)^2\,\middle|\,\mathcal F_{\ell-1}\right]
\approx
\bar M_{\ell,i}^f+\frac{(\bar M_{\ell,i}^f)^2}{r}+\Big(1+\frac1r\Big)V_{\ell,i}^f .
\label{eq:r_identity}
\end{equation}
The error of \eqref{eq:r_identity} is the error of this ensemble approximation. The term \((1+1/r)V^f_{\ell,i}\) is the contribution of forecast uncertainty to the expected squared residual. Omitting it would attribute forecast spread to observation overdispersion and would therefore tend to underestimate \(r\), particularly when the forecast ensemble is broad relative to its mean. Solving \eqref{eq:r_identity} for \(r\) and summing numerators and denominators over a window \(\mathcal W_k\) of recent forecast--observation pairs gives the forecast-spread-adjusted rolling estimator \(N_k/D_k\) with
\begin{equation}
N_k=\sum_{(\ell,i)\in\mathcal W_k}\big[(\bar M_{\ell,i}^f)^2+V_{\ell,i}^f\big],
\qquad
D_k=\sum_{(\ell,i)\in\mathcal W_k}\big[(Y_{\ell,i}-\bar M_{\ell,i}^f)^2-\bar M_{\ell,i}^f-V_{\ell,i}^f\big].
\label{eq:r_estimator}
\end{equation}
When \(V^f\equiv0\), this reduces to the usual residual-moment estimator. The adjustment uses only quantities already computed by the filter, but its accuracy depends on how well the ensemble mean and variance reproduce the predictive moments of \(M_{\ell,i}\). It therefore corrects for forecast uncertainty only to the extent that the ensemble approximation is reliable and does not establish that the resulting ratio estimator is unbiased. The estimate is \(\widehat r_k=r_{\max}\) if \(D_k\leq0\) and \(\widehat r_k=\clip(N_k/D_k,r_{\min},r_{\max})\) otherwise, and it is updated on the log scale,
\begin{equation}
\log r_k=(1-\eta_r)\log r_{k-1}+\eta_r\log\widehat r_k ,
\label{eq:r_update}
\end{equation}
with a small \(\eta_r\in(0,1)\). The update starts only after a burn-in fraction \(f_r\) of the record. Since \(N_k/D_k\) is a ratio of sums of squared residuals, it weights observations approximately as \((\bar M^f)^2\) and is a noisy, potentially heavy-tailed statistic in strongly overdispersed regimes. The delay, window and damping limit this effect, and a trimmed or Huberised residual moment is a possible refinement.

\subsection{Algorithm summary}\label{sec:algorithm}

Algorithm~\ref{alg:nbenpgf} states the complete procedure. For every node \(i\), each assimilation step first applies the parameter forecast. This is the gamma jitter \eqref{eq:jitter} except for row \(i\) of \(A\) when the composition forecast \eqref{eq:composition} is enabled. It then projects the parameters onto \eqref{eq:domain}, propagates \(\lambda_i\) through \eqref{eq:model}, transforms to \(u\), and moment-matches a beta forecast with the concentration floor \eqref{eq:concentration_floor}. Draws from \eqref{eq:beta_update} are assigned by rank using \eqref{eq:rank}. The parameters \(\mu_i\), \(\beta_i\) and \(\alpha_{ij}\) are updated by \eqref{eq:regression} and the result is projected onto \eqref{eq:domain}. After all nodes have been processed, the stored forecast means and variances enter the scheduled update of \(r\) by \eqref{eq:r_estimator}. Every member must satisfy \eqref{eq:domain}, and the beta approximation must satisfy \eqref{eq:bf_condition} through \eqref{eq:concentration_floor} with \(b_\star(r)\) recomputed whenever \(r\) changes. The composition forecast uses the zero-row convention of Section~\ref{sec:adaptability}. When it is disabled the gamma jitter \eqref{eq:jitter} is applied to \(A\).

The analysis update has three principal approximations: the beta forecast is obtained by moment matching, the posterior draws are coupled to the forecast members by rank, and the parameter update \eqref{eq:regression} is linear. The parameter forecast introduces an additional adaptation--stability trade-off because larger perturbations improve responsiveness but increase variation when the network is constant.

\begin{algorithm}[t]
\caption{Negative-binomial ensemble filter with optional composition forecast}
\label{alg:nbenpgf}
\begin{algorithmic}[1]
\Require Counts \(Y_{k,i}\), ensemble size \(\Ne\), observation interval \(\Delta t\), initial ensembles of \((\lambda_i,\mu_i,\beta_i,\alpha_{ij})\), initial dispersion \(r_0\), tolerance \(\epsilon_b\), jitter scale \(\sigma_\zeta\), dispersion parameters \((f_r,\eta_r,r_{\min},r_{\max})\) and, optionally, \((q_A,p_{\min})\).
\State \(r\leftarrow r_0\).
\For{\(k=0,1,\ldots,n-1\)}
    \For{\(i=1,\ldots,m\)}
        \State Multiply \(\mu_i\) and \(\beta_i\) of each member by the gamma jitter \eqref{eq:jitter}.
        \If{composition forecast enabled} \State Evolve row \(i\) of each \(A^{(s)}\) by \eqref{eq:composition}. \Else \State Apply the gamma jitter \eqref{eq:jitter} to row \(i\) of each \(A^{(s)}\). \EndIf
        \State Project onto \eqref{eq:domain}.
        \State \(\lambda_{k+1,i}^{f,(s)}\leftarrow\) \eqref{eq:model} with member parameters and \(Y_k\).
        \State Set \(M^{f,(s)}=\lambda^{f,(s)}_{k+1,i}\Delta t\) and \(u^{f,(s)}=M^{f,(s)}/(r+M^{f,(s)})\).
        \State Moment-match \(\Beta(a_f,b_f)\) by \eqref{eq:concentration_floor} with \(b_\star(r)\).
        \State Observe \(y=Y_{k+1,i}\) and draw \(\widetilde u^{a,(s)}\sim\Beta(a_f+y,b_f+r)\).
        \State Assign by rank \eqref{eq:rank} and map back to \(M^{a,(s)},\lambda^{a,(s)}\).
        \State Update \(\mu_i,\beta_i,\alpha_{ij}\) by \eqref{eq:regression}.
        \State Project onto \eqref{eq:domain}.
    \EndFor
    \State Store \(\bar M^f\), \(V^f\) and \(Y\) for dispersion learning.
    \If{\(k\geq\lfloor f_rn\rfloor\) and an \(r\)-update is scheduled}
        \State Form \(N_k,D_k\) from \eqref{eq:r_estimator}.
        \State \(\widehat r\leftarrow r_{\max}\) if \(D_k\leq0\), else \(\widehat r\leftarrow\clip(N_k/D_k,r_{\min},r_{\max})\).
        \State \(\log r\leftarrow(1-\eta_r)\log r+\eta_r\log\widehat r\).
        \State Recompute \(b_\star(r)\).
    \EndIf
\EndFor
\State \Return Ensembles of \(\lambda_k,\mu,\beta,A_k\) and the trajectory of \(r_k\).
\end{algorithmic}
\end{algorithm}

\section{Numerical experiments}\label{sec:experiments}

\subsection{Experimental design and evaluation}\label{sec:design}

In this section, two twin experiments are carried out in which the data-generating model and the inference model share the structure \eqref{eq:model}--\eqref{eq:nb}, so the true excitation matrix is known at every time. Counts are simulated through the gamma--Poisson mixture representation of \eqref{eq:nb}, and both experiments use the implementation of Section~\ref{sec:algorithm}. Twenty paired replications are run in each experiment, and in each replication every filter receives identical counts, an identical initial ensemble and the same starting random state.

The static experiment has \(m=6\), \(n=10000\) assimilation steps, \(\Ne=1000\), \(\Delta t=1\) and a constant truth \(A^\dagger\). Baselines and decays are heterogeneous, \(\mu_i^\dagger\sim\mathcal U(0.8,1.2)\) and \(\beta_i^\dagger\sim\mathcal U(0.45,0.65)\) independently, and the common dispersion is \(r^\dagger=0.5\). One admissible truth is drawn once and held fixed across replications, subject to \(\spr(G)<1\) and \(\spr(L)<0.97\) in Proposition~\ref{prop:stability}, its realised radii being \(0.9075\) and \(0.9398\), respectively. The stationary mean count is \(6.07\), so the conditional Fano factor \(1+M/r^\dagger\) is \(13.1\) and the information ratio \(r^\dagger/(r^\dagger+M)\) of Section~\ref{sec:model_spec} is \(0.076\), so each count carries about a thirteenth of the information that a Poisson count of the same mean would. Thus, the regime is strongly overdispersed. Two filters are run on this design. The first is the negative-binomial ensemble filter of Algorithm~\ref{alg:nbenpgf}, hereafter the NB filter. The second is the EnPGF in the form applied to Hawkes network reconstruction \citep{santitissadeekorn2023physd}. It starts from the shared initial ensemble, which the paired design requires. Both share the intensity forecast \eqref{eq:model} and the parameter forecast, namely the mean-one gamma jitter \eqref{eq:jitter} of the EnPGF applied to \(\mu\), \(\beta\) and \(A\) at every step with relative standard deviation \(\sigma_\zeta=0.32\%\). With the forecast shared, the two filters differ only in the analysis, so the comparison is between two complete filters and a difference between them cannot be attributed to the observation likelihood alone. The NB filter is initialised at the misspecified \(r_0=2.0\) and learns the dispersion by \eqref{eq:r_estimator} with the forecast-spread term retained. The dispersion update starts after the first \(10\%\) of the record, uses a rolling window of 500 observations and is applied every five steps with \(\eta_r=0.03\). The composition forecast is disabled.

The tracking experiment is designed to show that the composition forecast of Section~\ref{sec:adaptability}, which keeps the forecast variation of every edge large enough, adapts to a change in the network better than the gamma jitter. It has \(m=6\), \(n=20000\), \(\Ne=1000\), \(\Delta t=1\), \(r^\dagger=0.5\), \(r_0=2.0\) and a piecewise-constant truth, \(A_k^\dagger=A_{\mathrm{pre}}^\dagger\) for \(k\leq k_c\) and \(A_{\mathrm{post}}^\dagger\) for \(k>k_c\) with \(k_c=10000\). Neither the change point nor the affected entries are supplied to the filter. In each of the two destination rows carrying the largest off-diagonal incoming edge, that edge loses \(85\%\) of its value, which is added to the weakest off-diagonal edge of the same row, all other entries unchanged. The change preserves row totals and the row-wise bounds \eqref{eq:first_order} and \eqref{eq:second_order_row}, although it need not preserve the stationary mean. It moves \(\alpha_{42}\) from \(0.225\) to \(0.034\) and \(\alpha_{64}\) from \(0.217\) to \(0.033\), and it moves \(\alpha_{41}\) and \(\alpha_{61}\), which are exactly zero before the change, to \(0.191\) and \(0.185\), so the network goes from eighteen to twenty nonzero entries with \(\|A_{\mathrm{post}}^\dagger-A_{\mathrm{pre}}^\dagger\|_F/\|A_{\mathrm{pre}}^\dagger\|_F=0.565\). Both truths satisfy \(\spr(G)<1\) and \(\spr(L)<0.97\), the realised second-moment radii being \(0.9543\) before the change and \(0.6856\) after it. The two edges that are exactly zero before the change are deliberate, since an edge that the analysis updates have reduced to zero is the case in which the bound \eqref{eq:collapse} is tightest. Two configurations of the NB filter are run on this design and differ only in the forecast model applied to \(A\). The first uses the composition forecast \eqref{eq:composition} at every step with \(q_A=0.03\) and \(p_{\min}=10^{-3}\). The second replaces it by the mean-one gamma jitter of the static experiment at the same relative standard deviation of \(0.32\%\). All other settings are common to the two configurations, including the counts, the initial ensemble, the starting random state, the intensity forecast, the regression transform, the jitter applied to \(\mu\) and \(\beta\) and the dispersion schedule of the static experiment, so any difference between the two configurations is due to the forecast model applied to \(A\).

Let \(\bar A_k\) be the ensemble mean of the excitation matrix after assimilating counts up to time \(k\), taken after the projection onto \eqref{eq:domain} and before the next parameter forecast, and \(A_k^\dagger\) the truth active at \(k\). Network recovery is measured by the relative Frobenius error \(e_k=\|\bar A_k-A_k^\dagger\|_F/\|A_k^\dagger\|_F\). For the static experiment, \(\widehat A\) is the time average of \(\bar A_k\) over the second half of the record, \(\mathcal T=\{\lfloor n/2\rfloor,\ldots,n\}\), and \(e(\widehat A)=\|\widehat A-A^\dagger\|_F/\|A^\dagger\|_F\). The tail mean of the dispersion trajectory is defined by the same averaging. Every summary is computed within a replication and then averaged over the twenty replications, and every contrast between filters is a paired difference across replications reported with a \(95\%\) \(t\) interval on \(19\) degrees of freedom. Since all replications of an experiment share one truth, these intervals reflect sampling variability in the counts and the initial ensembles and not variability across networks. Time averaging removes step-to-step fluctuation, so \(e(\widehat A)\) is generally below \(e_n\), and averaging \(\widehat A\) across replications removes a further sampling component, so the error of the replication-averaged matrices shown in the figures is lower still.

Three further quantities summarise the static experiment for each filter, all evaluated on the forecast ensemble before assimilation and averaged over \(\mathcal T\). Parameter calibration is measured by the inclusion rate, the fraction of edges in \(\mathcal E=\{(i,j):\alpha^\dagger_{ij}>0\}\) whose true value lies between the \(5\%\) and \(95\%\) member quantiles, nominal \(0.90\). Edges with \(\alpha^\dagger_{ij}=0\) are excluded from \(\mathcal E\) since every member is nonnegative after the projection onto \eqref{eq:domain}, so a true value of zero can never lie inside the member interval and such an edge would count as a miss for every filter. The ensemble spread is compared with the error of the ensemble mean over the same edges by
\begin{equation}
\mathrm{SS}=\left\{\frac{(1+\Ne^{-1})\sum_{k\in\mathcal T}\sum_{(i,j)\in\mathcal E}\operatorname{sd}_k(\alpha_{ij})^2}{\sum_{k\in\mathcal T}\sum_{(i,j)\in\mathcal E}(\bar\alpha_{k,ij}-\alpha^\dagger_{ij})^2}\right\}^{1/2},
\label{eq:spread_skill}
\end{equation}
in which \(\operatorname{sd}_k\) is the memberwise standard deviation at time \(k\) and the factor \(1+\Ne^{-1}\) corrects for comparing member spread with the error of a mean of \(\Ne\) exchangeable members. A value of one is the calibrated case, below one an over-confident ensemble and above one an over-dispersed one. Predictive calibration is assessed on the one-step-ahead predictive distribution of each count, the ensemble mixture with mass function \(\Ne^{-1}\sum_{s}p(y\mid M^{f,(s)}_{k,i})\), where \(p\) is the negative-binomial mass function at the current dispersion estimate for the NB filter and the Poisson mass function for the EnPGF. It is measured by the log score of this mixture at the observed counts,
\begin{equation}
S_k=\frac1m\sum_{i=1}^m\log\Big\{\frac1{\Ne}\sum_{s=1}^{\Ne}p\big(Y_{k,i}\mid M^{f,(s)}_{k,i}\big)\Big\},
\label{eq:mixture_log_score}
\end{equation}
by the \(5\%\)--\(95\%\) coverage of the same mixture, nominal \(0.90\), and by its randomised probability integral transform \citep[PIT,][]{czado2009predictive}. The PIT mean is compared with \(0.5\), and the PIT variance divided by the uniform value \(1/12\) gives the variance ratio, whose nominal value is one. The mixture form is used rather than a plug-in score at \(\bar M^f\) because the ensemble spread is what is being measured. These three quantities need no knowledge of the truth and are the diagnostics used in the application of Section~\ref{sec:chicago}.

For the tracking experiment, the pre- and post-change estimates \(\widehat A_{\mathrm{pre}}\) and \(\widehat A_{\mathrm{post}}\) average \(\bar A_k\) over the final \(15\%\) of the pre-change segment and the final \(35\%\) of the post-change segment, their errors to the truth active in each window are \(e_{\mathrm{pre}}\) and \(e_{\mathrm{post}}\), and recovery of the change is measured by
\begin{equation}
e_{\Delta A}=\frac{\|(\widehat A_{\mathrm{post}}-\widehat A_{\mathrm{pre}})-(A_{\mathrm{post}}^\dagger-A_{\mathrm{pre}}^\dagger)\|_F}{\|A_{\mathrm{post}}^\dagger-A_{\mathrm{pre}}^\dagger\|_F}.
\label{eq:change_error}
\end{equation}
The inclusion rate and the spread-to-skill ratio \eqref{eq:spread_skill} are computed over the same two windows with \(\mathcal E\) taken from the truth active in the window, and the predictive diagnostics over the whole record.

\subsection{Results of the static experiment}\label{sec:static_results}

The initial ensemble mean is effectively flat (Figure~\ref{fig:static_networks}). Its entries lie between \(0.085\) and \(0.086\), and its relative Frobenius error is \(0.687\). The NB filter reaches \(e(\widehat A)=0.124\) against \(0.154\) for the EnPGF (Table~\ref{tab:static}), a paired difference of \(-0.030\) with interval \([-0.040,-0.019]\). Panel (a) of Figure~\ref{fig:enpgf_comparison} shows the error \(e_k\) of the NB filter below that of the EnPGF for most of the record. In the replication-averaged estimate of Figure~\ref{fig:static_networks}, whose error is \(0.057\), moderate and strong edges are recovered accurately and the twelve structural zeros stay below \(0.014\). Neither filter resolves the weakest edge, whose coefficient of \(0.015\) carries too little information at \(r^\dagger=0.5\) to be distinguished from absence in ten thousand steps. Similar point accuracy is consistent with quasi-likelihood results for count time series, under which the Poisson score remains a valid estimating function when the conditional mean is correctly specified \citep{fokianos2009poisson,christou2014quasi,ahmad2016qmle,davis2016theory}, even though those results concern batch estimation rather than sequential filtering.

The dispersion is learned from the misspecified \(r_0=2.0\). The median estimate over replications first falls below \(0.60\) at step 1480, and the mean of the second-half estimates across replications is \(0.553\) with a standard deviation of \(0.021\), against \(r^\dagger=0.5\).

The calibration diagnostics of the two filters differ by an order of magnitude in the direction anticipated in Section~\ref{sec:intro}. Panel (b) of Figure~\ref{fig:enpgf_comparison} shows that the inclusion rate of the EnPGF falls below \(0.5\) within the first few hundred steps. Its mean over the second half of the record is \(0.445\), while that of the NB filter is \(0.898\) against a nominal \(0.90\). The spread-to-skill ratio \eqref{eq:spread_skill} is \(0.372\) against \(1.112\). The EnPGF's ensemble is thus about one third as wide as the error of its own mean, while the NB ensemble is close to calibrated. Panels (c) and (d) of Figure~\ref{fig:enpgf_comparison} show the corresponding difference in the one-step predictive distributions. The log score is \(-2.551\) against \(-5.893\), the predictive coverage is \(0.950\) against \(0.538\), and the PIT of the EnPGF is biased low and much too narrow. Both predictive distributions are ensemble mixtures, but their variances differ. With \(\bar M^f\) and \(V^f\) the forecast mean and variance of Section~\ref{sec:dispersion}, the Poisson mixture of the EnPGF and the negative-binomial mixture of the NB filter have variances
\begin{equation}
\begin{aligned}
\Var(Y)&=\bar M^f+V^f,\\
\Var(Y)&=\E\!\left[M^f+\frac{(M^f)^2}{r}\right]+V^f=\bar M^f+\frac{(\bar M^f)^2}{r}+\Big(1+\frac1r\Big)V^f,
\end{aligned}
\label{eq:mixture_var}
\end{equation}
the second being the exact counterpart for the finite mixture of the ensemble approximation \eqref{eq:r_identity}. The Poisson mixture omits both the conditional overdispersion term and its contribution from variation among the member means. At \(r^\dagger=0.5\) it can reach the width the data require only through \(V^f\), and the EnPGF's contracted ensemble makes that term small as well. Its intervals are therefore much too narrow, and its log score is dominated by the large counts to which it assigns very small probability. The NB coverage of \(0.950\) is above the nominal \(0.90\) because the counts are discrete. The PIT, which accounts for this discreteness, has mean \(0.491\) and variance ratio \(1.030\) (Table~\ref{tab:static}), close to the values \(0.5\) and \(1\) of a calibrated predictive distribution. The lower NB coverage over the first thousand steps in panel (d) of Figure~\ref{fig:enpgf_comparison} reflects the predictive distribution at \(r_0=2.0\), which is too narrow, before the dispersion update starts.

The direction of the parameter-spread result is the one predicted by the information ratio \(r/(r+M)\) of Section~\ref{sec:model_spec}, which is \(0.076\) here. The Poisson analysis removes more information from the intensity ensemble at each step than the counts carry, and the regression passes the excess contraction to the parameters. The jitter common to both filters replenishes spread at the same rate for each, and the two filters therefore retain different amounts of it. The early decrease of the inclusion rate in panel (b) of Figure~\ref{fig:enpgf_comparison} is common to both filters and is the finite-ensemble underdispersion described in Section~\ref{sec:adaptability}. The subsequent recovery of the NB filter, which the EnPGF does not share, reflects their distinct observation-analysis constructions.

\begin{figure}[t]
    \centering
    \IfFileExists{static3_networks.pdf}{%
        \includegraphics[width=\textwidth]{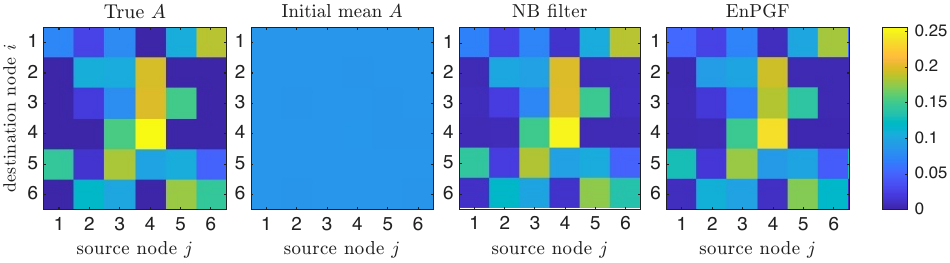}%
    }{%
        \fbox{\begin{minipage}[c][0.18\textheight][c]{0.92\textwidth}
        \centering Figure file \texttt{static3\_networks.pdf} not found.
        \end{minipage}}%
    }
    \caption{Static six-node twin experiment (\(n=10000\), \(\Ne=1000\), \(r^\dagger=0.5\), \(r_0=2.0\)). The excitation matrices share a common colour scale and show the truth \(A^\dagger\), the initial ensemble mean, and the estimates \(\widehat A\) for the two filters. Each estimate is averaged over the second half of the record and over the twenty paired replications, so its relative Frobenius error, \(0.057\) for the NB filter, is below the per-replication value reported in Table~\ref{tab:static}.}
    \label{fig:static_networks}
\end{figure}

\begin{table}[t]
\centering
\caption{Static experiment. Entries are means over replications with standard deviations in parentheses.}
\label{tab:static}
\footnotesize
\setlength{\tabcolsep}{5pt}
\begin{tabular}{lrrr}
\toprule
 & NB filter & EnPGF & NB \(-\) EnPGF \([95\%\ \mathrm{CI}]\) \\
\midrule
\(e(\widehat A)\)                                 & 0.124 (0.019) & 0.154 (0.015) & \(-0.030\ [-0.040,-0.019]\) \\
\(\mathrm{SS}\) \eqref{eq:spread_skill}           & 1.112 (0.132) & 0.372 (0.027) & \(+0.740\ [+0.682,+0.799]\) \\
inclusion rate                                    & 0.898 (0.045) & 0.445 (0.043) & \(+0.453\ [+0.432,+0.473]\) \\
log score \eqref{eq:mixture_log_score}            & \(-2.551\) (0.036) & \(-5.893\) (0.291) & \(+3.342\ [+3.222,+3.463]\) \\
predictive coverage                               & 0.950 (0.002) & 0.538 (0.011) & \(+0.412\ [+0.407,+0.417]\) \\
PIT mean                                          & 0.491 (0.004) & 0.354 (0.001) & \(+0.137\ [+0.136,+0.139]\) \\
PIT variance ratio                                & 1.030 (0.011) & 1.873 (0.013) & \(-0.843\ [-0.850,-0.837]\) \\
\(\widehat r\)                                    & 0.553 (0.021) & --- & --- \\
\bottomrule
\end{tabular}
\end{table}

\begin{figure}[t]
    \centering
    \IfFileExists{static3_metrics.pdf}{%
        \includegraphics[width=\textwidth]{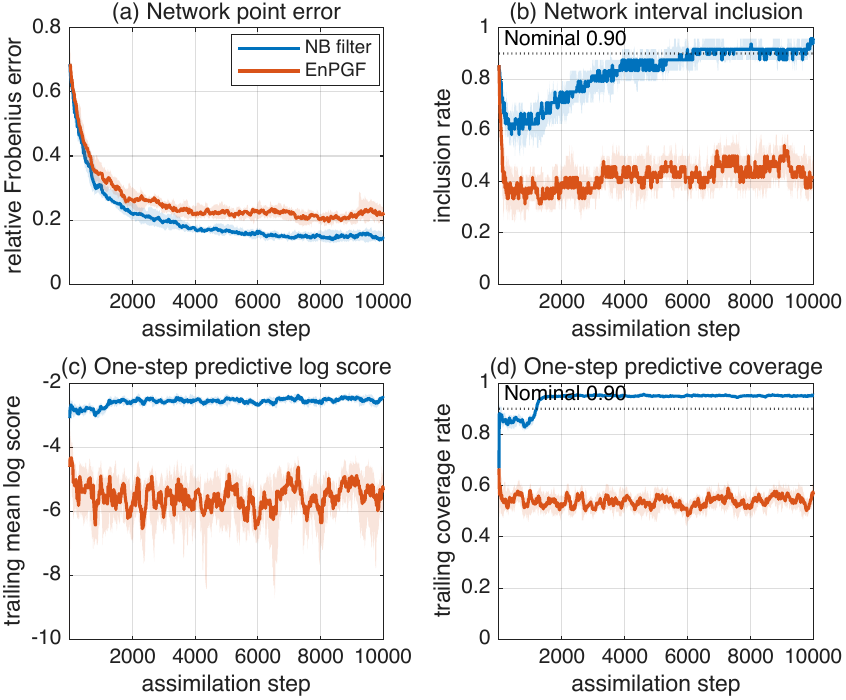}%
    }{%
        \fbox{\begin{minipage}[c][0.42\textheight][c]{0.92\textwidth}
        \centering Figure file \texttt{static3\_metrics.pdf} not found.
        \end{minipage}}%
    }
    \caption{Static experiment comparing the NB filter with the EnPGF. Curves show medians over the twenty replications with interquartile bands. (a) Relative Frobenius error \(e_k\) of the ensemble-mean excitation matrix. (b) Fraction of nonzero edges whose true coefficient lies in the \(5\%\)--\(95\%\) ensemble interval (nominal \(0.90\)). (c) One-step-ahead ensemble-mixture predictive log score \eqref{eq:mixture_log_score}, trailing mean, larger is better. (d) One-step-ahead \(5\%\)--\(95\%\) predictive coverage, trailing mean, nominal \(0.90\).}
    \label{fig:enpgf_comparison}
\end{figure}

\subsection{Results of the tracking experiment}\label{sec:tracking_results}

Before the change the gamma jitter is the more accurate of the two configurations (Table~\ref{tab:tracking}). Its error \(e_{\mathrm{pre}}\) is \(0.154\) against \(0.231\), a paired difference of \(+0.077\) with interval \([+0.060,+0.093]\). Section~\ref{sec:adaptability} gives the reason. Diffusing the composition of each row at every step prevents the ensemble from concentrating on the single network that generated the segment, and the analysis update opposes that diffusion without removing it. This raises the ensemble mean at the structural zeros, to \(0.015\) against \(0.010\) for the jitter, and adds step-to-step variation. After the change the composition forecast is more accurate. Its error \(e_{\mathrm{post}}\) is \(0.185\) against \(0.255\), a difference of \(-0.070\) with interval \([-0.097,-0.044]\), and its error in the change matrix is \(e_{\Delta A}=0.453\) against \(0.543\), a difference of \(-0.090\) with interval \([-0.131,-0.049]\). Figure~\ref{fig:tracking_networks} shows the corresponding windowed estimates averaged over replications. Rows four to seven of Table~\ref{tab:tracking} show both ensembles close to nominal before the change. After it the jitter ensemble contracts onto the network it estimated before the change, its spread-to-skill ratio falling to \(0.633\) and its inclusion rate to \(0.840\), while the composition ensemble stays at \(1.299\) and \(0.954\) and alone admits further update through \eqref{eq:regression}. Panel (a) of Figure~\ref{fig:tracking} shows the instantaneous error \(e_k\) through the change, which ends lower for the jitter, at \(0.248\) against \(0.345\). The composition trajectory fluctuates more from step to step but is centred closer to the truth, which is what the windowed errors measure.

The difference between the two configurations lies at the edges that the change activates. Panels (b) and (d) of Figure~\ref{fig:tracking} show the two entries that the change reduces to near zero, and both configurations follow their decline. Panels (c) and (e) show the two entries that the change activates from zero. At the end of the record the composition estimates are \(0.190\) and \(0.177\) against true values of \(0.191\) and \(0.185\), whereas the jitter estimates are \(0.106\) and \(0.105\). At the change point the composition forecast also maintains substantially wider member intervals at these two entries. These wider intervals are the forecast variation supplied by the composition forecast of Section~\ref{sec:adaptability}, which by the bound \eqref{eq:collapse} allows the regression update to move the two entries away from zero. The jitter perturbs each entry in proportion to its current value and therefore restores almost no forecast variance at an entry near zero.

The two configurations have nearly identical one-step predictive diagnostics (rows eight to eleven of Table~\ref{tab:tracking}) since the change redistributes influence while preserving row totals and a one-step count forecast is determined mainly by the baseline and the persistence term of \eqref{eq:model}. Predictive calibration alone therefore cannot validate recovery of the network, a restriction that applies to the Chicago application of Section~\ref{sec:chicago}, where no reference network exists.

\begin{table}[t]
\centering
\caption{Tracking experiment. Entries are means over replications with standard deviations in parentheses.}
\label{tab:tracking}
\footnotesize
\setlength{\tabcolsep}{5pt}
\begin{tabular}{lrrr}
\toprule
 & composition forecast & gamma jitter & difference \([95\%\ \mathrm{CI}]\) \\
\midrule
\(e_{\mathrm{pre}}\)                          & 0.231 (0.040) & 0.154 (0.030) & \(+0.077\ [+0.060,+0.093]\) \\
\(e_{\mathrm{post}}\)                         & 0.185 (0.030) & 0.255 (0.052) & \(-0.070\ [-0.097,-0.044]\) \\
\(e_{\Delta A}\) \eqref{eq:change_error}      & 0.453 (0.068) & 0.543 (0.073) & \(-0.090\ [-0.131,-0.049]\) \\
inclusion rate, pre                           & 0.929 (0.031) & 0.920 (0.050) & \(+0.009\ [-0.011,+0.029]\) \\
inclusion rate, post                          & 0.954 (0.017) & 0.840 (0.034) & \(+0.115\ [+0.100,+0.130]\) \\
\(\mathrm{SS}\) \eqref{eq:spread_skill}, pre  & 1.255 (0.110) & 1.179 (0.221) & \(+0.075\ [-0.030,+0.181]\) \\
\(\mathrm{SS}\), post                         & 1.299 (0.085) & 0.633 (0.146) & \(+0.666\ [+0.579,+0.752]\) \\
log score \eqref{eq:mixture_log_score}        & \(-2.189\) (0.018) & \(-2.185\) (0.018) & \(-0.004\ [-0.005,-0.003]\) \\
predictive coverage                           & 0.951 (0.002) & 0.953 (0.002) & \(-0.002\ [-0.003,-0.002]\) \\
PIT mean                                      & 0.488 (0.003) & 0.497 (0.004) & \(-0.008\ [-0.010,-0.007]\) \\
PIT variance ratio                            & 1.033 (0.011) & 1.013 (0.012) & \(+0.021\ [+0.018,+0.023]\) \\
\(\widehat r\), post                          & 0.574 (0.021) & 0.520 (0.022) & \(+0.054\ [+0.048,+0.059]\) \\
\bottomrule
\end{tabular}
\end{table}

\begin{figure}[!tbp]
    \centering
    \IfFileExists{tracking2_networks.pdf}{%
        \includegraphics[width=\textwidth]{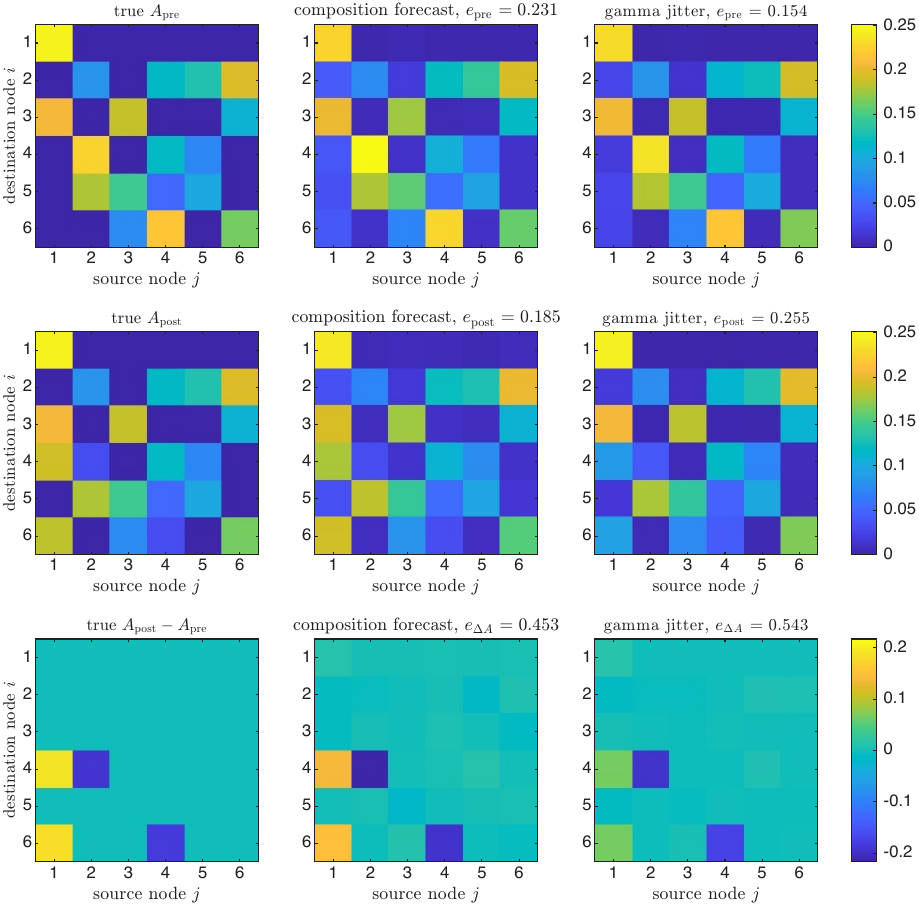}%
    }{%
        \fbox{\begin{minipage}[c][0.55\textheight][c]{0.92\textwidth}
        \centering Figure file \texttt{tracking2\_networks.pdf} not found.
        \end{minipage}}%
    }
    \caption{Tracking experiment, windowed network estimates averaged over replications. Columns give the truth, the composition forecast and the gamma jitter. Rows give the pre-change network, the post-change network and the change matrix. The six network panels share one colour scale and the three change panels a symmetric one. The errors in the titles are the replication means of Table~\ref{tab:tracking} rather than the error of the averaged matrix displayed, which is lower.}
    \label{fig:tracking_networks}
\end{figure}

\begin{figure}[!tbp]
    \centering
    \IfFileExists{tracking2.pdf}{%
        \includegraphics[width=\textwidth]{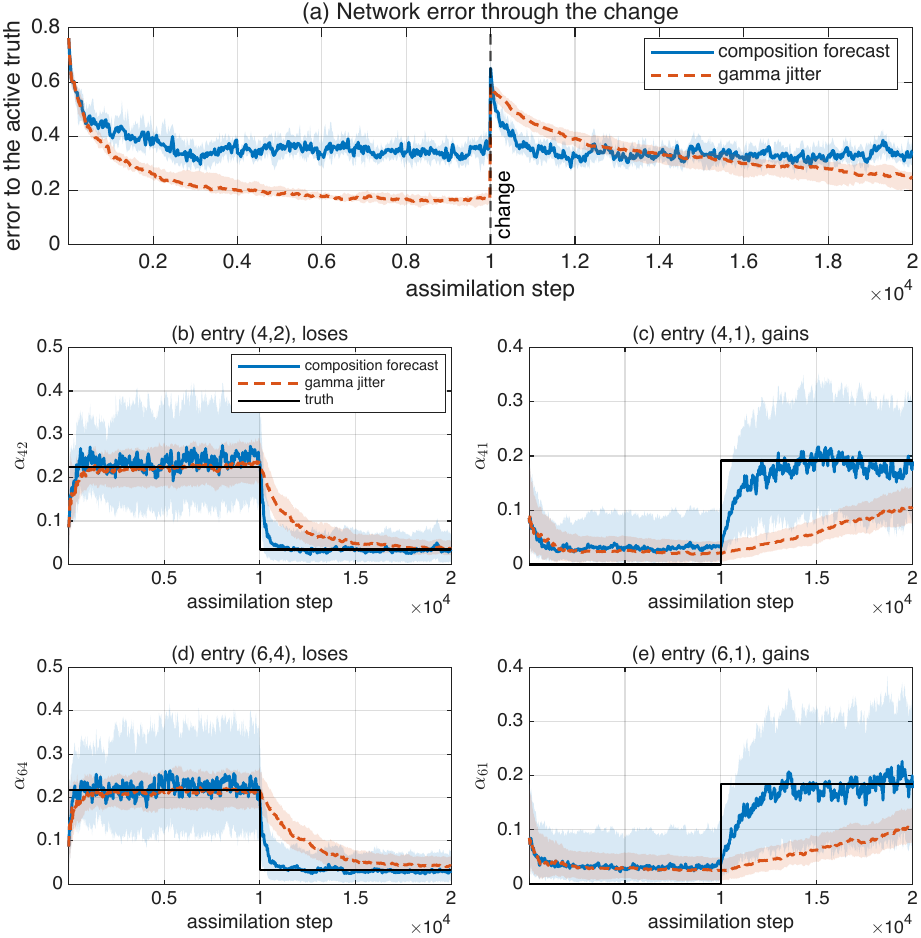}%
    }{%
        \fbox{\begin{minipage}[c][0.60\textheight][c]{0.92\textwidth}
        \centering Figure file \texttt{tracking2.pdf} not found.
        \end{minipage}}%
    }
    \caption{Tracking experiment with twenty paired replications. Panel (a) shows the relative Frobenius error \(e_k\) of the ensemble-mean network to the currently active truth, the median over replications with an interquartile band, and a vertical rule at the unannounced change. Panels (b)--(e) show the four modified entries of \(A\). Each gives the piecewise-constant truth, the median over replications of the ensemble mean, and the median over replications of the within-run \(5\%\)--\(95\%\) member interval. Entries \(\alpha_{41}\) and \(\alpha_{61}\) are exactly zero before the change, and the jitter configuration estimates them only slowly because its perturbation is proportional to the coefficient it perturbs.}
    \label{fig:tracking}
\end{figure}

\clearpage

\section{Application to daily property-crime counts in Chicago}\label{sec:chicago}

This section applies the NB filter to daily crime counts from Chicago. Chicago crime records have been used before to study excitation between the community areas of the city, for gang-related homicides \citep{linderman2014discovering} and for burglaries \citep{mark2019network}. Here the application concerns daily counts of property crime in sixteen contiguous community areas over 2021--2024. Property crime, and burglary in particular, is known to show near-repeat victimisation, in which an offence raises the short-term risk of further offences nearby. This clustering has motivated self-exciting point-process models of crime \citep{mohler2011crime,reinhart2018review}. Daily counts of this kind are also overdispersed and seasonal, which makes them a natural test of the negative-binomial observation model in a setting where no reference network is available.

\subsection{Data and node set}\label{sec:chicago_data}

Incident-level crime records for the City of Chicago are published with a timestamp, a primary offence type and the community area in which the incident occurred \citep{chicagocrimes2026}. We take the four offence types that form the property index of the Uniform Crime Reports, theft, burglary, motor vehicle theft and arson, and aggregate them to daily counts by community area, so \(\Delta t=1\) day and the discrete-time recursion \eqref{eq:model} is the model rather than an approximation to a continuous-time process. The filter is run on 1 January 2021 to 31 December 2024, which gives \(n=1461\) assimilation steps. Records from 2015 and 2016 are used only to fit the climatology of Section~\ref{sec:chicago_clim}. Area-days with no incident are entered as zero and make up 4.3\% of the panel.

The node set is the \(m=16\) community areas shown in Figure~\ref{fig:chicago_map}, a West Side block (areas 22--31) and the lakefront--downtown strip (areas 3, 6, 7, 8, 32 and 33). Their boundaries are taken from the city's open data portal \citep{chicagoboundaries2026}. They form one contiguous study region, which reduces the potential influence of omitted intermediate areas, and they carry counts high enough for daily modelling. Over the window they record 143{,}159 incidents \citep{chicagocrimes2026}. Mean daily counts, shown in Figure~\ref{fig:chicago_map}(a), range from 2.1 in West Garfield Park (26) to 12.9 in the Near North Side (8), and the overall mean is 6.1.

\begin{figure}[t]
    \centering
    \IfFileExists{map_pair_property.pdf}{%
        \includegraphics[width=\textwidth]{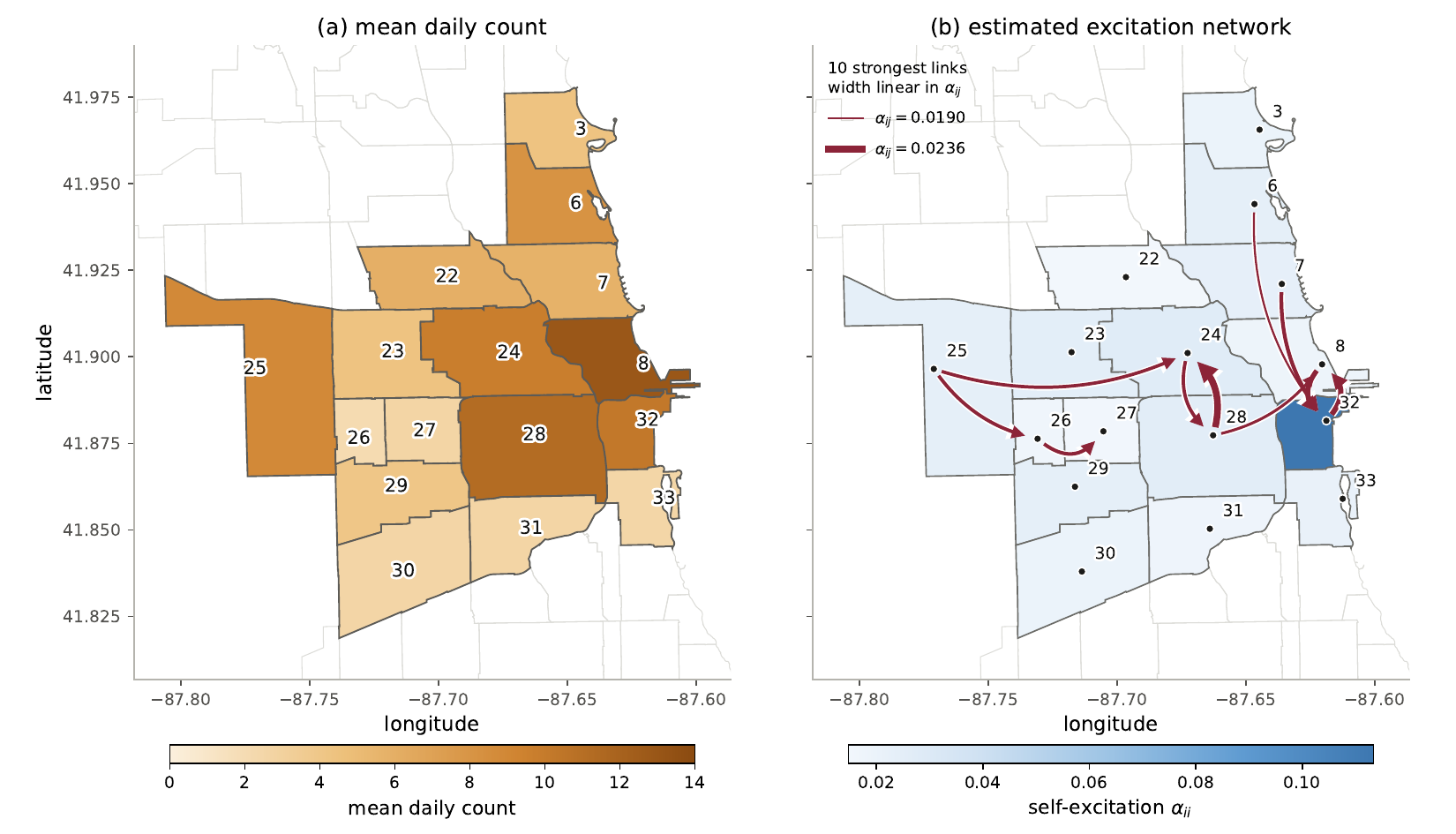}%
    }{%
        \fbox{\begin{minipage}[c][0.3\textheight][c]{0.92\textwidth}
        \centering Figure file \texttt{map\_pair\_property.pdf} not found.
        \end{minipage}}%
    }
    \caption{The sixteen modelled community areas, 1 January 2021 to 31 December 2024, labelled by community-area number and referred to by that number throughout. The rest of the city is outlined for context. (a) Mean daily property-crime count. (b) The excitation network recovered by the filter. Each area is filled by its estimated self-excitation \(\widehat\alpha_{ii}\), and the ten largest off-diagonal entries are drawn as arrows from source \(j\) to destination \(i\) with width linear in \(\widehat\alpha_{ij}\). Seven of the ten join areas that share a boundary. The Loop (32) has a self-excitation of 0.113, about four times that of any other area, and the largest cross-excitation runs from the Near West Side (28) to West Town (24).}
    \label{fig:chicago_map}
\end{figure}

\subsection{Offline ``climatology'' and the anomaly recursion}\label{sec:chicago_clim}

The baseline \(\mu_i\) of \eqref{eq:model} is a constant, whereas the counts carry an annual cycle with a peak in mid-August about 52\% above the trough in mid-February, and a weekly cycle that differs between areas. With a constant baseline the filter can follow the seasonal level only through the excitation term because \(\sum_j\alpha_{ij}Y_{k,j}\) is the only input that varies with the calendar, and since every area rises and falls together the attribution enters every row of \(A\) in the same direction. We therefore estimate the deterministic component offline, calling it the ``climatology'' of the panel, and run the filter on the anomaly.

The climatology is a Poisson log-linear regression,
\begin{equation}
\log\E[Y_{k,i}] = c_i + d_{i,\mathrm{dow}(k)} + \sum_{h=1}^{3}\Big[a_h\sin\frac{2\pi h\,t_k}{365.2425} + b_h\cos\frac{2\pi h\,t_k}{365.2425}\Big],
\label{eq:climatology}
\end{equation}
with \(t_k\) the day count from a fixed origin, six weekday effects \(d_{i,\cdot}\) per node, and annual harmonics \(a_h,b_h\) pooled across nodes. The annual cycle of property crime has one shape citywide and pooling estimates it from sixteen series rather than one, whereas the weekday cycle differs between areas, its maximum-to-minimum ratio being 1.22 at the median area and 1.51 in the Loop (32). The regression is fitted on 1 January 2015 to 31 December 2016 and frozen, so the whole of the analysis window is out of sample with respect to every seasonal parameter. The seasonal multiplier
\begin{equation}
\begin{aligned}
s_{k,i}&=\frac{g_{k,i}}{|\mathcal K_0|^{-1}\sum_{\ell\in\mathcal K_0}g_{\ell,i}},\\
g_{k,i}&=\exp\Big(d_{i,\mathrm{dow}(k)}+\sum_{h=1}^{3}\Big[a_h\sin\frac{2\pi h\,t_k}{365.2425}+b_h\cos\frac{2\pi h\,t_k}{365.2425}\Big]\Big),
\end{aligned}
\label{eq:seasonal_multiplier}
\end{equation}
where \(\mathcal K_0\) is the set of days in the training window and \(g_{k,i}\) is the seasonal part of \eqref{eq:climatology}, has arithmetic mean one over the training window for every node, so \(\mu_i\) keeps its meaning as the average baseline level, and slow drift in the level remains for the filter to track through \(\mu_i\). Figure~\ref{fig:chicago_clim} shows the frozen multiplier over the counts it is applied to for the Near North Side (8). The filter is then run with \(\mu_i s_{k,i}\) in place of \(\mu_i\),
\begin{equation}
\lambda_{k+1,i}=\mu_i s_{k+1,i}+(1-\beta_i\Delta t)\big(\lambda_{k,i}-\mu_i s_{k,i}\big)+\sum_{j=1}^m\alpha_{ij}Y_{k,j},
\label{eq:anomaly_recursion}
\end{equation}
which is \eqref{eq:model} applied to the anomaly \(\lambda_{k,i}-\mu_i s_{k,i}\). The counts are not transformed, the observation model \eqref{eq:nb} and the update of Section~\ref{sec:analysis} are unchanged, and \(s\equiv1\) recovers the constant-baseline filter exactly.

\begin{figure}[t]
    \centering
    \IfFileExists{climatology_fit_property.pdf}{%
        \includegraphics[width=\textwidth]{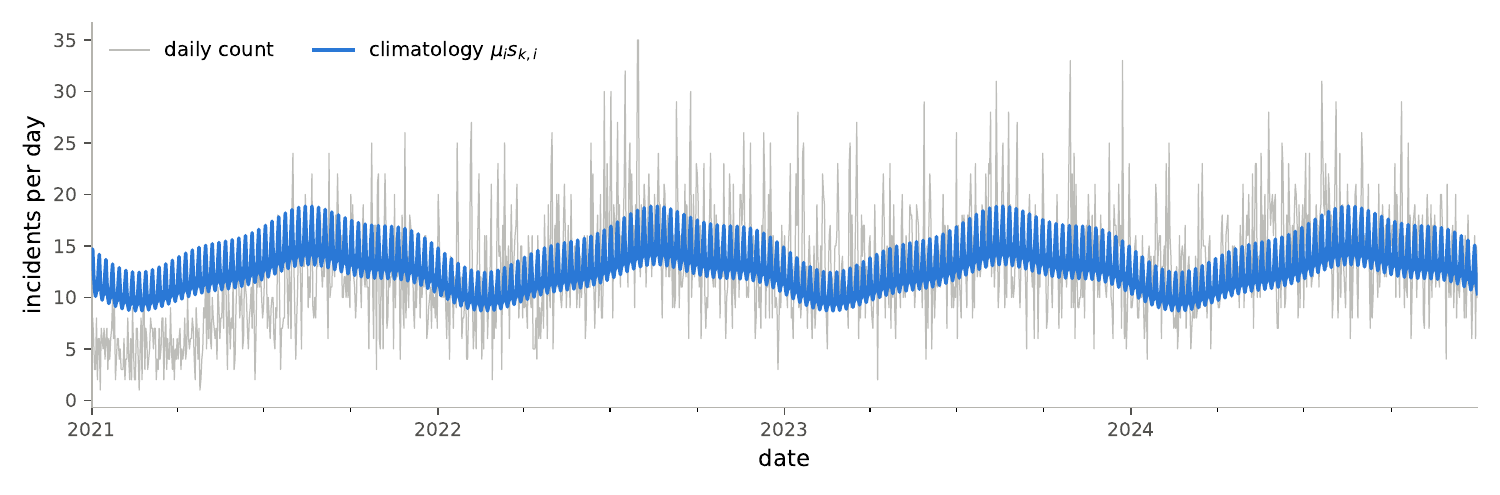}%
    }{%
        \fbox{\begin{minipage}[c][0.2\textheight][c]{0.92\textwidth}
        \centering Figure file \texttt{climatology\_fit\_property.pdf} not found.
        \end{minipage}}%
    }
    \caption{Daily property-crime counts in the Near North Side (8) over the analysis window. The climatology \(\mu_i s_{k,i}\) of \eqref{eq:climatology}--\eqref{eq:seasonal_multiplier} is superimposed. The seasonal shape was fitted on 2015--2016 and frozen, so the curve is extrapolated from calendar position alone. Here \(\mu_i\) is fixed at the mean of \(Y/s\) over the window, whereas in the filter it is a tracked state. The annual component and the weekly cycle are both visible.}
    \label{fig:chicago_clim}
\end{figure}

The climatology removes the component of dependence between consecutive days that is explained by their sharing a season or a part of the week, which is a function of the calendar and not of the previous day's counts. It cannot remove a burst that falls on particular days in one year and different days in the next because the fit is a function of the day of the year and the day of the week alone. It has 118 parameters against 11{,}696 training observations. A burst of that kind is represented as excitation in \eqref{eq:model}, and the residual lag-one autocorrelation over the analysis window is the clustering that remains for the filter.

Three quantities computed from the residuals of the seasonal regression over the analysis window, none of which is supplied to the filter, serve below as external checks. The lag-one autocorrelation of the residuals has median 0.23 across areas, reaching 0.65 in the Loop (32). Their lag-one cross-correlation between adjacent areas is 1.44 times that between non-adjacent areas, and their mean lag-one autocorrelation is 2.11 times their mean lag-one cross-correlation between distinct areas.

\subsection{Initialisation and online estimation}\label{sec:chicago_init}

The runs use the implementation of Section~\ref{sec:design} with the composition forecast disabled and \(\Ne=4000\).

The baseline is centred node by node from the first 365 days of the record after the seasonal multiplier is divided out. Writing \(\bar m_i\) for the mean of \(Y_{k,i}/s_{k,i}\) over those days, \(\bar m\) for its average over nodes, \(a_0\) for the initial row total and \(\beta_0=0.55\) for the initial decay, the stationary relation implied by \eqref{eq:count_mean} under a uniform row gives the centre
\begin{equation}
\mu_i^{0}=\max\Big\{\bar m_i-\frac{a_0\bar m}{\beta_0},\ 0.25\,\bar m_i\Big\}.
\label{eq:chicago_mu0}
\end{equation}
With \(a_0=0.20\) the floor is active for West Garfield Park (26), East Garfield Park (27), South Lawndale (30), Lower West Side (31) and Near South Side (33), so the nominal stationary-mean matching is truncated at those five areas to prevent a near-zero initial baseline. The matching is nominal. The \(\mu\) ensemble is log-normal about this centre with spread 0.10, kept small because a diffuse \(\mu\)-ensemble contributes a variance to \(M^f\) that dominates the covariance between \(\alpha_{ij}\) and \(M^f\) on which \eqref{eq:regression} depends. Each \(\alpha_{ij}\) is log-normal with median \(a_0/m\) and log-standard deviation 0.50. We take \(a_0=0.20\), for which the ensemble-mean row total is \(0.227\) rather than \(a_0\) and excitation carries about 23\% of the conditional mean at initialisation. The decay \(\beta_i\) is log-normal with centre 0.55 and log-standard deviation 0.25. Each member then passes a stability screen, which shrinks its row of \(A\) until its row sum of \(G\), \((1-\beta_i\Delta t)+\Delta t\sum_j\alpha_{ij}\), lies below \(0.97\) and thereby enforces the row-sum bound \eqref{eq:first_order}.

The dispersion is estimated online by the forecast-spread-adjusted estimator \eqref{eq:r_estimator} with one change from the twin experiments. There the sums run over every node and every step in the window and return one \(r\) for the network. Here they run over steps only and return one value per node,
\begin{equation}
\hat r_i=\frac{\sum_{k}\big[(\bar M^{f}_{k,i})^{2}+V^{f}_{k,i}\big]}{\sum_{k}\big[(Y_{k,i}-\bar M^{f}_{k,i})^{2}-\bar M^{f}_{k,i}-V^{f}_{k,i}\big]},
\label{eq:r_node}
\end{equation}
which the observation model \eqref{eq:nb} accommodates without change. The pooled sum is dominated by the largest squared residuals, which belong to the most overdispersed areas, and the single \(r\) it returns leaves the predictive distribution too wide in the less-dispersed areas, as Section~\ref{sec:chicago_results} quantifies. The window is 200 days rather than the 500 of Section~\ref{sec:design} to allow faster adaptation of the dispersion within the record, and the log-space update rate is raised correspondingly to 0.12. Each \(\hat r_i\) starts at 25, is updated every five steps from 5\% of the record, and is clamped to \([0.5,500]\). The uniform variates that randomise the probability integral transform below come from a stream separate from the filter's. The remaining constants are as in Section~\ref{sec:design}.

\subsection{Evaluation without a reference network}\label{sec:chicago_eval}

The evaluation windows are as follows. The first 365 days, calendar year 2021, initialise \(\mu\) and enter no statistic. Predictive diagnostics are computed before each assimilation from 2022 onward and are reported by calendar year, and the summary statistics of Table~\ref{tab:chicago}, together with the average that defines \(\widehat A\), are taken over the second half of the record, 2023--2024. Predictive calibration is measured first by the 5\%--95\% coverage of the predictive distribution of Section~\ref{sec:design}, the fraction of observations \(Y_{k,i}\) that fall within the central 90\% interval of its distribution function \(F_{k,i}(y)=\Ne^{-1}\sum_s F_{\mathrm{NB}}\big(y;M^{f,(s)}_{k,i},\hat r_i\big)\), with \(F_{\mathrm{NB}}\) the negative-binomial distribution function, so the interval is that of the mixture and not of the member means. Coverage above nominal is expected for a count interval with integer endpoints, so it is supplemented by the randomised PIT. With \(F\) the mixture predictive distribution function and \(u\) uniform on \((0,1)\), \(\mathrm{PIT}_k=F(Y_k-1)+u\,[F(Y_k)-F(Y_k-1)]\) is uniform under calibration. Its variance relative to the uniform value \(1/12\) separates discreteness from over-dispersion, which coverage alone cannot, a ratio below one meaning the predictive distribution is too wide, and it is computed for the panel and for each area. Its mean by calendar month tests for seasonal bias.

The recovered network is assessed against information the filter does not receive. The share of the ten largest off-diagonal entries of \(\widehat A\) that join areas sharing a boundary is compared with the 25\% of ordered pairs that are adjacent, and the recovered self-excitation \(\widehat\alpha_{ii}\) is correlated across nodes with the lag-one autocorrelation of the seasonal-regression residuals. Neither measures error. Each is a condition a spatially meaningful excitation network would satisfy.

\subsection{Results}\label{sec:chicago_results}

Table~\ref{tab:chicago} reports the filter of Section~\ref{sec:chicago_init} beside two runs that differ from it in one setting each, the constant baseline \(s\equiv1\) in place of the climatology and a single pooled \(r\) in place of \eqref{eq:r_node}.

\begin{table}[t]
\centering
\caption{The negative-binomial ensemble filter on daily property-crime counts in sixteen Chicago community areas, 1 January 2021 to 31 December 2024, \(n=1461\), \(\Ne=4000\). The first column is the configuration of Section~\ref{sec:chicago_init} and the other two change one setting each. Statistics are averaged over the second half of the record, which is also the window over which \(\widehat A\) is averaged, and \(\widehat A_1\) and \(\widehat A_2\) are the averages over its first and second halves. The share of off-diagonal excitation in the ten largest entries would be \(10/240=0.042\) under a uniform matrix, and 25\% of ordered pairs of areas are adjacent.}
\label{tab:chicago}
\footnotesize
\setlength{\tabcolsep}{4pt}
\begin{tabular}{lccc}
\toprule
 & climatology, node \(r\) & \(s\equiv1\) & pooled \(r\) \\
\midrule
dispersion \(\hat r\), median [range]                        & 40 [13,\,205] & 36 [13,\,209] & 19 \\
\(\spr(B+\Delta t\,\widehat A)\)                              & 0.897 & 0.926 & 0.887 \\
mean row total \(\sum_j\widehat\alpha_{ij}\) (initial 0.227)  & 0.215 & 0.215 & 0.218 \\
\midrule
PIT mean (0.5 under calibration)                             & 0.496 & 0.493 & 0.500 \\
PIT variance ratio (1 under calibration)                     & 0.953 & 0.953 & 0.915 \\
\quad by year, 2022 / 2023 / 2024                            & .94 / .94 / .97 & .95 / .94 / .97 & .79 / .88 / .95 \\
areas with PIT variance ratio in \([0.9,1.1]\)               & 14/16 & 14/16 & 8/16 \\
5\%--95\% coverage (nominal 0.90)                            & 0.943 & 0.944 & 0.952 \\
corr(monthly count, monthly PIT mean)                        & \(-0.55\) & \(+0.76\) & \(-0.61\) \\
\midrule
share of off-diagonal excitation in ten largest entries      & 0.070 & 0.068 & 0.064 \\
adjacent share of ten largest entries                        & 7/10 & 7/10 & 7/10 \\
adjacent / non-adjacent mean \(\widehat\alpha_{ij}\)         & 1.24 & 1.20 & 1.20 \\
corr\((\widehat\alpha_{ij},\ \text{centroid distance})\)      & \(-0.33\) & \(-0.28\) & \(-0.34\) \\
diagonal / off-diagonal of \(\widehat A\)                    & 2.17 & 2.04 & 2.11 \\
corr\((\widehat\alpha_{ii},\ \text{residual acf}_1)\)        & 0.87 & 0.86 & 0.86 \\
\midrule
tail-median step increment (\(\times10^{-3}\))               & 4.2 & 4.4 & 3.8 \\
split-half \(\|\widehat A_1-\widehat A_2\|_F/\|\widehat A\|_F\) & 0.126 & 0.118 & 0.114 \\
\bottomrule
\end{tabular}
\end{table}

The calibration rows of Table~\ref{tab:chicago} compare the predictive distribution with the values it would give under correct specification. The PIT mean is 0.496 against 0.5 and the PIT variance ratio is 0.953 against 1, where a value below one means the predictive distribution is too wide. Coverage is 0.943 against the nominal 0.90, the excess reflecting the discreteness of the counts as in Section~\ref{sec:static_results}. The variance ratio lies within 10\% of one in fourteen of the sixteen areas, the two exceptions both being too wide, and between 0.94 and 0.97 in each evaluation year. The predictive distribution is therefore calibrated in mean and in width, overall and marginally by area and by year.

The pooled-\(r\) column shows what the node-specific estimator \eqref{eq:r_node} contributes. The estimates \(\hat r_i\) range from 13 to 205, which is a Fano factor \(1+\bar M_i/\hat r_i\) from 1.01 to 1.62, and their ordering is unrelated to the level of the counts. Pooling returns \(\hat r=19\) for every area, a value set by the most overdispersed areas, whose squared residuals dominate the pooled sum. The predictive distribution is then too wide in the less-dispersed areas. The variance ratio falls to 0.915, and to 0.79 in 2022, and only eight of the sixteen areas are within 10\% of calibration. The recovered network changes little. The two estimates of \(A\) differ by 0.109 in relative Frobenius norm, which is less than the split-half difference within one run, and they share seven of their ten largest entries.

The constant-baseline column shows what the climatology contributes. Without it the monthly mean PIT correlates with the monthly mean count at \(+0.76\), and the forecast is biased downward in summer and upward in winter. With it the correlation is \(-0.55\), smaller in magnitude and opposite in sign, since the frozen 2015--2016 amplitude overshoots the 2021--2024 cycle. A month-dependent bias therefore remains in both columns, and the calibration statements above are marginal. The width of the predictive distribution and the network rows are almost the same in both columns. The annual cycle that is absorbed by the excitation term under a constant baseline thus appears mainly as a seasonal bias of the predictive distribution and only weakly in the recovered network.

The network rows summarise the recovered network shown in Figure~\ref{fig:chicago_map}(b). The mean row total ends at 0.215 against an initial 0.227. Two further runs with \(a_0=0.05\) and \(a_0=0.40\), the reference configuration otherwise unchanged, begin at mean row totals of 0.057 and 0.427 and end at 0.058 and 0.327. The final row totals retain a strong dependence on their initial values across this sevenfold range, which indicates that the total excitation entering an area is weakly informed by these data. The counts inform the distribution of that total within a row more strongly. The diagonal carries 13\% of the total excitation on 6\% of the entries, and the largest self-excitation, in the Loop (32), is eight times its initial value. The off-diagonal entries are much more uniform. The ten largest cross-excitations carry 7.0\% of the off-diagonal excitation against 4.2\% under uniformity, and Figure~\ref{fig:chicago_map}(b) displays the largest entries of a nearly uniform matrix rather than a sparse network.

The ranking is nonetheless geographic although no geographic information is supplied to the filter. Seven of the ten largest cross-excitations in Figure~\ref{fig:chicago_map}(b) join areas that share a boundary, against the 25\% of ordered pairs that do, and the mean entry correlates with inter-centroid distance at \(-0.33\). The recovered structure also agrees with the residuals of the seasonal regression, which are not supplied to the filter either. The self-excitation correlates with the residual lag-one autocorrelation across areas at 0.87, and the diagonal-to-off-diagonal ratio of 2.17 is close to the 2.11 of the residuals, whereas the adjacent-to-non-adjacent ratio of 1.24 is below their 1.44. The filter therefore recovers the diagonal contrast in full and the geographic gradient in part. The Loop (32) has the strongest residual clustering and the largest recovered self-excitation, and it is also the one area whose predictive distribution remains clearly too wide. Its largest residuals fall on consecutive days at the end of July, when Grant Park within the area hosts a multi-day festival. The festival recurs every year but is too sharp to be represented by three annual harmonics, which makes the Loop a clear case in which unresolved calendar structure may be absorbed by the estimated excitation and dispersion. Whether such a burst belongs to the excitation term or to the background is a modelling choice rather than a discovery.

The last two rows of Table~\ref{tab:chicago} indicate that the estimates fluctuate around a stable level over the averaging window. The relative step-to-step increment of the ensemble-mean network falls from \(6.8\times10^{-3}\) over the first hundred steps to \(4.2\times10^{-3}\), which it holds for the second half of the record, and averaging the two halves of the window separately changes \(\widehat A\) by 0.126 in relative Frobenius norm.

Taken together, the rows describe a division in how strongly the counts inform the two aspects of \(\widehat A\). The level of excitation is weakly informed and stays close to its initial value, whereas the composition within a row is more strongly informed and reproducible between the two halves of the record, the diagonal more so than the off-diagonal, and where it is reproducible it is local. This is the division assumed by the composition forecast of Section~\ref{sec:adaptability}, and it motivates estimating the row total independently, for example by a batch method, while using the sequential filter to estimate its composition.
\section{Conclusion}\label{sec:conclusion}

We have introduced a negative-binomial ensemble filter for sequential reconstruction of the excitation network of a discrete-time self-exciting count process from overdispersed counts. The method combines a conjugate scalar update on the scale \(u=M/(r+M)\), a rank-preserving posterior sampler, an ensemble-regression transform from the intensity increment to the Hawkes parameters, a row-composition forecast that keeps the filter responsive to network changes, and online estimation of the dispersion, common to the network in the twin experiments and node-specific in the application. A concentration floor enforces the finite-variance condition \(b_f+r>2\) while preserving the transformed-scale mean, and the moment-stability conditions provide exact tests for the simulated count process.

Twin experiments illustrate static reconstruction of a six-node network under strong overdispersion, where twenty paired replications against the ensemble Poisson--gamma filter give comparable network accuracy and a predictive log score higher by \(3.34\ [3.22,3.46]\) and an interval inclusion rate of \(0.90\) against \(0.45\), and tracking of an abrupt, row-sum-preserving redistribution of influence, where twenty further paired replications that differ only in the parameter forecast show that the composition forecast reduces the post-change network error from \(0.255\) to \(0.185\) and recovers two edges that are exactly zero before the change, while the jitter leaves both at little more than half their true size. On daily property-crime counts in sixteen Chicago community areas over 2021--2024, with a climatology of annual harmonics and weekday effects estimated offline on 2015--2016 and the filter run on the anomaly, the filter gives one-step-ahead predictive distributions with a probability integral transform of mean 0.496 and variance 0.95 of the uniform value, calibrated marginally by area and by year with a month-dependent bias remaining, and a network whose ten largest cross-excitations join adjacent areas in seven cases against a chance rate of one in four although no geographic information is supplied to the filter. The same panel shows that a single dispersion parameter, set by the most overdispersed areas, leaves the predictive distribution too wide in the less-dispersed areas, and that estimating it node by node restores calibration while changing the recovered network by no more than the split-half variation within a single run. It also shows that on seasonal counts a constant baseline leaves a month-dependent bias in the one-step predictive distribution, which the frozen climatology reverses in sign and reduces, and that on such data the level of excitation is weakly informed by the counts, a sevenfold change in the initial row total being still visible at the end of the record, while its composition is informed more strongly.

The principal limitations are that the comparison with the ensemble Poisson--gamma filter covers the static case only and that the tracking experiment compares two parameter forecasts under a single observation model without an offline comparator such as a rolling-window likelihood built from \(\ell_{\mathrm{NB}}\). The factorised observation model does not represent contemporaneous dependence generated by shared latent drivers. A persistent common driver may consequently be absorbed into the estimated off-diagonal excitation. Extending the benchmark comparison to a changing truth, adding a static-versus-dynamic comparison on a constant truth, relaxing the cross-node independence assumption, and testing the node-specific dispersion estimator against a simulated truth with known heterogeneous dispersion are the immediate extensions.

\appendix
\numberwithin{equation}{section}
\numberwithin{figure}{section}

\section{Proof of Proposition~\ref{prop:stability}}\label{app:stability_proof}

We work with the count-mean recursion \eqref{eq:count_mean}, write \(\rho_i=1-\beta_i\Delta t\) for the diagonal entries of \(B\) so that \(0\leq\rho_i<1\) under \eqref{eq:domain}, and \(G=B+\mathcal A\), and use only the two conditional moments in \eqref{eq:nb}. The initial count mean \(\bm M_0\) is assumed nonnegative, so that \(\bm M_k\geq0\) for every \(k\) under \eqref{eq:domain}. The statements concern moment recursions and do not establish existence or uniqueness of a stationary law.

\subsection*{Part (i): mean}
Taking conditional expectations in \(\bm M_{k+1}=\bm\omega+B\bm M_k+\mathcal A\bm Y_k\) and using \(\E[\bm Y_k\mid\bm M_k]=\bm M_k\) gives \(\bm m_{k+1}=\bm\omega+G\bm m_k\) for \(\bm m_k=\E[\bm M_k]\). Hence \(\bm m_k\) converges from any finite initialisation if and only if \(\spr(G)<1\), with limit
\begin{equation}
    \bar{\bm m}=(I-G)^{-1}\bm\omega .
\label{eq:app_mean_limit}
\end{equation}
Note that the condition \(\spr(G)<1\) is necessary and this conclusion does not require \(G\) to be irreducible. To see this, let \(\bm v\geq0\), \(\bm v\neq0\) be a left Perron eigenvector of the nonnegative matrix \(G\), so \(\bm v^\top G=\spr(G)\bm v^\top\). Then \(\bm v^\top\bm\omega>0\) because \(\bm\omega>0\) entrywise, and pairing the recursion with \(\bm v\) gives \(\bm v^\top\bm m_{k+1}=\spr(G)\bm v^\top\bm m_k+\bm v^\top\bm\omega\), which diverges when \(\spr(G)\geq1\). A stationary solution with finite mean would need \((1-\spr(G))\bm v^\top\bar{\bm m}=\bm v^\top\bm\omega>0\), which is impossible when \(\spr(G)\geq1\). Since \(G\geq0\) entrywise, \(\spr(G)\leq\|G\|_\infty=\max_i(\rho_i+\sum_j\mathcal A_{ij})\), which is the row-sum bound \eqref{eq:first_order}. Hence, the row-sum bound is sufficient, but not necessary.

\subsection*{Part (ii): covariance}

\begin{proof}[Proof of Proposition~\ref{prop:stability}(ii)]
Let \(\Sigma_k=\Cov(\bm M_k)\). Conditional on \(\bm M_k\) the components of \(\bm Y_k\) are independent and \(\Cov(\bm Y_k\mid\bm M_k)=\diag(M_{k,i}+M_{k,i}^2/r)\), so \(\bm M_{k+1}\) has conditional mean \(\bm\omega+G\bm M_k\) and conditional covariance \(\mathcal A\diag(M_{k,i}+M_{k,i}^2/r)\mathcal A^\top\). The law of total covariance then gives
\begin{equation}
    \Sigma_{k+1}
    =
    G\Sigma_kG^\top
    +
    \mathcal A\,\diag\!\Big(m_{k,i}+\frac{\Sigma_{k,ii}+m_{k,i}^2}{r}\Big)\mathcal A^\top ,
    \label{eq:app_Sigma_recursion}
\end{equation}
which we write as \(\Sigma_{k+1}=\mathcal L(\Sigma_k)+C_k\) with
\begin{equation}
\begin{aligned}
    \mathcal L(X)&=GXG^\top+\tfrac1r\,\mathcal A\diag(X_{11},\ldots,X_{mm})\mathcal A^\top,\\
    C_k&=\mathcal A\,\diag\!\big(m_{k,i}+m_{k,i}^2/r\big)\mathcal A^\top .
\end{aligned}
\label{eq:app_L_C}
\end{equation}
The operator \(\mathcal L\) is linear in \(X\), and \(C_k\) depends only on the mean. Writing \(\vecop(X)\) for the columns of \(X\) stacked into a vector of length \(m^2\), \(\vecop(GXG^\top)=(G\otimes G)\vecop(X)\) and \(\vecop(\diag(X_{11},\ldots,X_{mm}))=EE^\top\vecop(X)\) with \(E\) as in Proposition~\ref{prop:stability}, so
\begin{equation}
    \vecop\!\big(\mathcal L(X)\big)=L\,\vecop(X),
    \qquad
    L=G\otimes G+\frac1r\,(\mathcal A\otimes\mathcal A)EE^\top ,
    \label{eq:app_L_operator}
\end{equation}
which is the matrix in \eqref{eq:second_order_exact}. Hence \(\spr(L)\) is the spectral radius of \(\mathcal L\).

We first prove sufficiency. Suppose that \(\spr(L)<1\). Since \(L\geq G\otimes G\) entrywise, \(\spr(L)\geq\spr(G)^2\), so \(\spr(G)<1\) and \(\bm m_k\) converges by part (i). Hence \(C_k\) converges, and the affine recursion with contracting linear part converges from any finite start.

We now prove necessity and divergence when \(\spr(L)\geq1\). The argument uses the uncentred second moment \(Q_k=\E[\bm M_k\bm M_k^\top]=\Sigma_k+\bm m_k\bm m_k^\top\), which is entrywise nonnegative since \(\bm M_k\geq0\). Two cases arise. Either \(\spr(G)<1\), so that \(\bm m_k\) is bounded by part (i) and divergence of \(Q_k\) implies divergence of \(\Sigma_k\), or \(\spr(G)\geq1\), so that \(Q_k\) can diverge through the mean alone and \(\Sigma_k\) must be treated directly. Adding \(\bm m_{k+1}\bm m_{k+1}^\top=(\bm\omega+G\bm m_k)(\bm\omega+G\bm m_k)^\top\) to \eqref{eq:app_Sigma_recursion} and using \(\Sigma_{k,ii}+m_{k,i}^2=Q_{k,ii}\) gives
\begin{equation}
    Q_{k+1}=\mathcal L(Q_k)+H_k,
    \qquad
    H_k=\bm\omega\bm\omega^\top+\bm\omega\bm m_k^\top G^\top+G\bm m_k\bm\omega^\top+\mathcal A\diag(\bm m_k)\mathcal A^\top ,
\label{eq:app_Q_recursion}
\end{equation}
with the same operator \(\mathcal L\) as in \eqref{eq:app_L_C}. Every term of \(H_k\) is entrywise nonnegative, and \(H_k\geq\bm\omega\bm\omega^\top>0\) entrywise since \(\bm\omega>0\). With \(\bm q_k=\vecop(Q_k)\) and \(\bm h_k=\vecop(H_k)\), \eqref{eq:app_L_operator} turns \eqref{eq:app_Q_recursion} into \(\bm q_{k+1}=L\bm q_k+\bm h_k\), in which \(L\), \(\bm q_k\) and \(\bm h_k\) are all entrywise nonnegative. By the Perron--Frobenius theorem for nonnegative matrices, \(L\) has a left eigenvector \(\bm w\geq0\), \(\bm w\neq0\), satisfying \(\bm w^\top L=\spr(L)\bm w^\top\), and no irreducibility is needed.

Consider first the case \(\spr(G)<1\). Put \(z_k=\bm w^\top\bm q_k\geq0\). Then
\begin{equation}
    z_{k+1}=\spr(L)\,z_k+\bm w^\top\bm h_k\;\geq\;z_k+c,
    \qquad
    c=\bm w^\top\vecop(\bm\omega\bm\omega^\top)>0 ,
\label{eq:app_z_w}
\end{equation}
so \(z_k\to\infty\), and since \(z_k\leq\|\bm w\|_2\|Q_k\|_F\) this gives \(\|Q_k\|_F\to\infty\). Since \(\spr(G)<1\), \(\bm m_k\) is bounded by part (i), and \(\|\Sigma_k\|_F\geq\|Q_k\|_F-\|\bm m_k\bm m_k^\top\|_F\) gives \(\|\Sigma_k\|_F\to\infty\). Since \(\bm\omega>0\), the forcing has a positive projection onto every nonzero nonnegative left Perron vector of \(L\). Consequently, the Perron component of \(Q_k\) receives a positive input at every step and diverges when \(\spr(L)\geq1\).

It remains to treat \(\spr(G)\geq1\), where the mean recursion is unstable. Let \(\bm v\geq0\), \(\bm v\neq0\), be a left Perron eigenvector of \(G\). Then \(\mathcal A^\top\bm v\neq0\). Otherwise \(\bm v^\top G=\bm v^\top B\), so \(v_i\rho_i=\spr(G)v_i\) for every \(i\), which forces \(v_i=0\) for all \(i\) because \(\rho_i<1\leq\spr(G)\). Put \(z_k=\bm v^\top\Sigma_k\bm v\geq0\). Pairing \eqref{eq:app_Sigma_recursion} with \(\bm v\) on both sides,
\begin{equation}
\begin{aligned}
    z_{k+1}
    &=\spr(G)^2z_k+\sum_{i}(\mathcal A^\top\bm v)_i^2\Big(m_{k,i}+\frac{\Sigma_{k,ii}+m_{k,i}^2}{r}\Big)\\
    &\geq\spr(G)^2z_k+\sum_{i}(\mathcal A^\top\bm v)_i^2\,\omega_i
    \;=:\;\spr(G)^2z_k+c ,
\end{aligned}
\label{eq:app_z_perron}
\end{equation}
using \(\Sigma_{k,ii}\geq0\) and \(m_{k,i}\geq\omega_i>0\) for \(k\geq1\), which follows from \(\bm m_{k+1}=\bm\omega+G\bm m_k\geq\bm\omega\) since \(\bm m_k\geq0\). Since \(c>0\) and \(\spr(G)^2\geq1\), \(z_k\to\infty\), hence \(\|\Sigma_k\|_F\to\infty\).

For the stationary claim, a strictly stationary solution with finite covariance has finite mean, so by part (i) it satisfies \(\spr(G)<1\). Its uncentred second moment \(Q=\Sigma+\bar{\bm m}\bar{\bm m}^\top\) is constant, and \eqref{eq:app_Q_recursion} gives \(\bm q=L\bm q+\bm h\) with \(\bm h\geq\vecop(\bm\omega\bm\omega^\top)\). Pairing with \(\bm w\), \((1-\spr(L))\,\bm w^\top\bm q=\bm w^\top\bm h>0\) with \(\bm w^\top\bm q\geq0\), which forces \(\spr(L)<1\).

Finally, for the row-wise test write \(g_i=\sum_aG_{ia}\) and \(s_i=\sum_l\mathcal A_{il}\), so that \eqref{eq:second_order_row} is
\begin{equation}
    \max_{i}\left\{g_i^{2}+\frac{s_i^{2}}{r}\right\}<1 .
    \label{eq:app_second_order_sufficient}
\end{equation}
Since \(L\geq0\), \(\spr(L)\leq\|L\|_\infty\). The row of \(L\) indexed by \((i,j)\) sums to
\begin{equation}
    \sum_{a,b}G_{ia}G_{jb}+\frac1r\sum_{l}\mathcal A_{il}\mathcal A_{jl}
    \;\leq\;
    g_ig_j+\frac{s_is_j}{r},
\label{eq:app_row_sum}
\end{equation}
using \(\sum_l\mathcal A_{il}\mathcal A_{jl}\leq(\sum_l\mathcal A_{il}^2)^{1/2}(\sum_l\mathcal A_{jl}^2)^{1/2}\leq s_is_j\) for nonnegative entries. Applying the Cauchy--Schwarz inequality in \(\R^2\) to the vectors \((g_i,s_i/\sqrt r)\) and \((g_j,s_j/\sqrt r)\),
\begin{equation}
    g_ig_j+\frac{s_is_j}{r}
    \;\leq\;
    \Big(g_i^2+\tfrac{s_i^2}{r}\Big)^{1/2}\Big(g_j^2+\tfrac{s_j^2}{r}\Big)^{1/2}
    \;\leq\;
    \max_{i}\Big\{g_i^2+\tfrac{s_i^2}{r}\Big\},
\label{eq:app_cs}
\end{equation}
so \eqref{eq:app_second_order_sufficient} implies \(\spr(L)<1\).
\end{proof}

Since \(\max_i\{g_i^2+s_i^2/r\}\leq\max_ig_i^2+r^{-1}\max_is_i^2=\|G\|_\infty^2+r^{-1}\|\mathcal A\|_\infty^2\), the left-hand side of \eqref{eq:app_second_order_sufficient} is an upper bound on \(\spr(L)\) that is never larger than the bound obtained by maximising the two terms separately, and the two bounds are equal when the same row attains both maxima. Condition \eqref{eq:app_second_order_sufficient} is therefore no more restrictive than \(\|G\|_\infty^2+r^{-1}\|\mathcal A\|_\infty^2<1\), and it still implies \(\spr(L)<1\). Unlike that bound, it is a per-row test.
In the univariate case \(L\) is the scalar \(\gamma^2+a^2/r\), so \eqref{eq:app_second_order_sufficient} coincides with \(\spr(L)<1\) and is exact. For \(m>1\) it is conservative, and \eqref{eq:app_L_operator} should be used when a sharp answer is required. Both conditions tighten monotonically as \(r\) decreases and both reduce to the corresponding Poisson conditions as \(r\to\infty\). If the dispersion parameter is node-specific, \(\tfrac1r EE^\top\) in \eqref{eq:app_L_operator} is replaced by \(ER^{-1}E^\top\) with \(R=\diag(r_1,\ldots,r_m)\), and the Cauchy--Schwarz argument above, applied to \(\sum_l\mathcal A_{il}\mathcal A_{jl}/r_l\) without the relaxation to \(s_is_j\), gives the sufficient condition \(\max_i\{g_i^2+\sum_l\mathcal A_{il}^2/r_l\}<1\) in place of \eqref{eq:app_second_order_sufficient}.

\bibliographystyle{elsarticle-harv}
\bibliography{NBHawkes_refs}

\end{document}